\documentclass[sigconf,authorversion]{acmart}
\usepackage{multirow}
\usepackage{subfigure}
\usepackage[ruled,vlined,linesnumbered]{algorithm2e}
\newcommand{\stitle}[1]{\vspace*{0.4em}\noindent{\bf #1.\/}}

\newcommand{\revised}[1]{#1}
\newcommand{\squishlist}{
	\begin{list}{$\bullet$}
		{ \setlength{\itemsep}{1pt}
			\setlength{\parsep}{1pt}
			\setlength{\topsep}{2.5pt}
			\setlength{\partopsep}{0.5pt}
			\setlength{\leftmargin}{1em}
			\setlength{\labelwidth}{1em}
			\setlength{\labelsep}{0.6em}
		}
	}
	\newcommand{\squishend}{
	\end{list}
}

\newtheorem{remark}{Remark}
\AtBeginDocument{%
  }

\acmConference[SIGMOD ’26]{ACM on Management of Data}{May 31–June 5, 2026}{Bengaluru, India}
\setcopyright{cc}
\setcctype{by-sa}
\acmYear{2025} 
\acmVolume{3} \acmNumber{6 (SIGMOD)} \acmArticle{359} \acmMonth{12}
\acmDOI{10.1145/3769824}

\begin{document}

\title{SAQ: Pushing the Limits of Vector Quantization through Code Adjustment and Dimension Segmentation}

\author{Hui Li}
\affiliation{%
  \institution{The Chinese University of Hong Kong}
  \country{Hong Kong SAR}
}
\orcid{0009-0009-0772-6588}
\email{hli@cse.cuhk.edu.hk}

\author{Shiyuan Deng}
\affiliation{%
  \institution{Huawei Cloud}
  \country{China}
}
\orcid{0009-0008-6088-6313}
\email{dengshiyuan@huawei.com}

\author{Xiao Yan}
\authornote{Dr. Xiao Yan is the corresponding author.}
\affiliation{%
  \institution{Institute for Math \& AI, Wuhan, Wuhan University}
  \country{China}
}
\orcid{0000-0002-2122-915X}
\email{yanxiaosunny@whu.edu.cn}

\author{Xiangyu Zhi}
\affiliation{%
  \institution{The Chinese University of Hong Kong}
  \country{Hong Kong SAR}
}
\orcid{0009-0001-0122-9240}
\email{xyzhi24@cse.cuhk.edu.hk}

\author{James Cheng}
\affiliation{%
  \institution{The Chinese University of Hong Kong}
  \country{Hong Kong SAR}
}
\orcid{0000-0001-6313-6288}
\email{jcheng@cse.cuhk.edu.hk}






\begin{abstract}

Approximate Nearest Neighbor Search (ANNS) plays a critical role in applications such as search engines, recommender systems, and RAG for LLMs. Vector quantization (VQ), a crucial technique for ANNS, is commonly used to reduce space overhead and accelerate distance computations. However, despite significant research advances, state-of-the-art VQ methods still face challenges in balancing encoding efficiency and quantization accuracy. To address these limitations, we propose a novel VQ method called SAQ. To improve accuracy, SAQ employs a new dimension segmentation technique to strategically partition PCA-projected vectors into segments along their dimensions. By prioritizing leading dimension segments with larger magnitudes, SAQ allocates more bits to high-impact segments, optimizing the use of the available space quota. An efficient dynamic programming algorithm is developed to optimize dimension segmentation and bit allocation, ensuring minimal quantization error. To speed up vector encoding, SAQ devises a code adjustment technique to first quantize each dimension independently and then progressively refine quantized vectors using a coordinate-descent-like approach to avoid exhaustive enumeration. Extensive experiments demonstrate SAQ's superiority over classical methods (e.g., PQ, PCA) and recent state-of-the-art approaches (e.g., LVQ, Extended RabitQ). SAQ achieves up to 80\% reduction in quantization error and accelerates encoding speed by over 80× compared to Extended RabitQ. 

\end{abstract}





\maketitle

\begingroup\small\noindent\raggedright\textbf{Artifact Availability:}\\
The source code, data, and/or other artifacts have been made available at \url{https://github.com/howarlii/saq/}.
\endgroup

\section{Introduction} \label{sec:intro}
With the proliferation of machine learning, embedding models~\cite{devlin2019bert, radford2021learning, shvetsova2022everything, li2022competition} are widely used to map diverse data objects, including images, videos, texts, e-commerce goods,  genes, and proteins, to high-dimensional vector embeddings that encode their semantic information. A core operation on these vectors  is \textit{nearest neighbor search} (\textit{NNS})~\cite{Elastic, pgvector, wang2021milvus, SingleStore, Qdrant, yang2020pase, mohoney2023high}, which retrieves vectors that are the most similar to a query vector from a vector dataset. NNS underpins many important applications including content search (e.g., for images and videos), recommender systems, bio-informatics, and retrieval-argumented generation (RAG) for LLMs~\cite{lewis2020retrieval}.
However, exact NNS requires a linear scan in high-dimensional space, rendering it impractical for large-scale datasets. To address this, \textit{approximate NNS} (\textit{ANNS}) has been widely adopted~\cite{huang2020embedding, xiong2020approximate}, which trades exactness for efficiency by returning most of the top-$k$ nearest neighbors. Many vector indexes~\cite{malkov2018efficient, yin2025gorgeousrevisitingdatalayout} and vector databases~\cite{Elastic, pgvector, wang2021milvus, SingleStore, Qdrant} support ANNS as the core functionality.


Vector quantization (or vector compression) maps each vector $\mathbf{o}$ to a smaller approximate vector $\bar{\mathbf{o}}$ and is a key technique for efficient ANNS. With billions of vectors and each vector having thousands of dimensions, vector datasets can be large (e.g., in TBs) and vector quantization helps reduce space consumption. Moreover, vector quantization  enables us to compute approximate distance (i.e., $\Vert \mathbf{q}-\bar{\mathbf{o}} \Vert$, where $\mathbf{q}$ is the query vector) faster than  the exact distance (i.e., $\Vert \mathbf{q}-\mathbf{o} \Vert$). Since distance computation dominates the running time of vector indexes~\cite{gao2023highdimensionalapproximatenearestneighbor}, vector quantization effectively accelerates ANNS.
The goal of vector quantization is to minimize the \textit{relative error} (defined as $\big \lvert \Vert \mathbf{q}-\bar{\mathbf{o}} \Vert^2-\Vert \mathbf{q}-\mathbf{o} \Vert^2 \big \rvert /\Vert \mathbf{q}-\mathbf{o} \Vert^2$) under a given space quota (i.e., measured by compression ratio w.r.t. the original vector or the average number of bits used for each dimension) such that  approximate distance matches exact distance.

\begin{figure}[!t]
  \centering
  \includegraphics[width=0.85\linewidth]{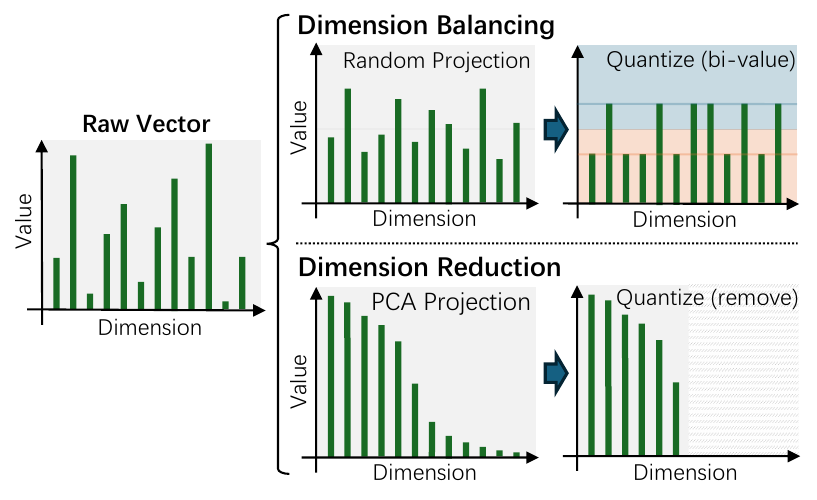}
  \caption{Illustration of dimension balancing and dimension reduction. Bar height is the magnitude of vector dimension.}
  \label{fig:alpha-fig}
  \vspace{-6mm}
\end{figure}

Many vector quantization methods have been proposed.
PQ~\cite{babenko2014additive, jegou2010product, martinez2018lsq++, wang2017survey}, OPQ~\cite{ge2013optimized}, and LOPQ~\cite{kalantidis2014locally} partition the $D$-dimensional vector space into subspaces and use K-means to learn vector codebooks for each subspace. 
AQ~\cite{babenko2014additive}, RQ~\cite{chen2010approximate}, LSQ~\cite{martinez2016revisiting, martinez2018lsq++}, and TreeQ~\cite{foote2003treeq, charami2007performance} further improve the structures and learning methods of the vector codebooks. Some other methods learn query-dependent codebooks to better approximate the query-vector distance~\cite{guo2016quantization, liu2016query}. The state-of-the-art vector quantization methods adopt two types of designs as illustrated in Figure~\ref{fig:alpha-fig}. The first one is~\textit{dimension balancing}, which uses a random orthonormal matrix $R$ to project a vector $x$ and then quantizes each dimension of $Rx$ to an integer. The representative is RaBitQ~\cite{rbq, extrbq}, which quantizes each projected vector dimension with 1 bit and provides unbiased distance estimation. The second type is~\textit{dimension reduction}, which uses a PCA projection matrix $P$ to project a vector $x$ such that the leading dimensions correspond to large eigenvalues and thus have large magnitudes. Then, the small tailing dimensions are discarded.

Despite the progresses, we observe two limitations in the state-of-the-art  vector quantization methods.
\squishlist
\item \textbf{High quantization complexity}: Currently, RaBitQ~\cite{rbq} and extended RaBitQ (E-RaBitQ)~\cite{extrbq} achieve the best accuracy. However, RaBitQ can only use 1 bit for each vector dimension, which limits accuracy, while the quantization complexity of E-RaBitQ is \(O(2^{B}\cdot D \log{D})\) where a vector has $D$ dimensions and each dimension uses $B$ bits. This is because E-RaBitQ does not handle each vector dimension independently, but an enumeration is required to decide the optimal quantized vector to approximate the original vector. Our profiling shows that when $B=9$ and $D=3,072$, quantizing 1 billion vectors with E-RaBitQ takes more than 3,600 CPU hours.

\item \textbf{Contradictory designs}: Dimension balancing attempts to make vector dimensions similar in magnitude such that the same number of bits can be used for each dimension. In contrast, dimension reduction makes the vector dimensions skewed in  magnitude so that a vector can be approximated by its leading dimensions. These two designs are fundamentally contradictory, posing problems about selecting the better method for different applications or considering synergistic potential for enhanced performance.
\squishend


The inherent conflict between the two dominant design paradigms in existing quantization methods makes it difficult to further improve the efficiency of quantization while retaining the accuracy of vector search. To resolve this challenge, we propose a novel vector quantization method, \textit{SAQ}, that overcomes the limitations of state-of-the-art methods through two key innovations: \textit{code adjustment} and \textit{dimension segmentation}. One striking contribution of SAQ is: SAQ significantly advances \textit{both} quantization efficiency and vector search accuracy of currently best methods. \revised{As shown in Figure~\ref{fig:error-tiny}, SAQ achieves significantly lower quantization errors than RaBitQ\footnote{In the subsequent discussion, we also refer to E-RaBitQ as RabitQ for simplicity as most of our discussions are related to E-RaBitQ.} with the same space quota, while delivering dramatic speed-ups of over 80x faster quantization.}

SAQ introduces a \textit{code adjustment} technique to accelerate vector quantization. In particular, after random orthonormal projection, SAQ first quantizes each dimension of a vector independently using $B$ bits. Observing that quantization accuracy depends on how well the original vector $x$ and quantized vector $\tilde{x}$ align in their directions (measured by cosine similarity), we propose to adjust the quantized code for each dimension of $\tilde{x}$ to better align with $x$. This yields a complexity of \(O( D)\) to quantize each vector, avoiding RaBitQ's code enumeration which requires  \(O(2^{B}\cdot D \log{D})\). We call the new quantization method \textit{CAQ} and show that CAQ achieves the same quantization accuracy as RaBitQ, as code adjustment essentially conducts coordinate descent~\cite{wright2015coordinate} to optimize  $\tilde{x}$.

SAQ utilizes a \textit{dimension segmentation} technique to bridge dimension reduction and dimension balancing for improved accuracy.
After PCA projection, the vector dimensions are partitioned into segments, and CAQ is applied independently to each dimension segment. The key is to allocate more bits to the segments of leading (high-magnitude) dimensions to improve accuracy, while segments of trailing dimensions use fewer bits or are discarded.
We develop a dynamic programming algorithm to optimally allocate bits across  dimension segments to minimize quantization error under a total space quota $B\cdot D$. We show that dimension segmentation ensures that SAQ's approximation error cannot be larger than RaBitQ's.

We further enhance SAQ with a set of novel optimizations. First, we design a distance estimators based on the quantized vectors and develop single-instruction-multiple-data (SIMD) implementations for efficient distance computation. Second,  CAQ supports \textit{progressive distance approximation}, enabling us to take the first $b<B$ bits for each dimension while still yielding a valid quantized vector (with reduced accuracy). Third, SAQ  supports multi-stage distance estimation, which gives lower and upper bounds for distances using the segmented approximate vectors. The two properties are useful for pruning in ANNS~\cite{gao2023highdimensionalapproximatenearestneighbor, extrbq}.

We conduct extensive experiments to evaluate SAQ and compare it with four representative vector quantization methods. The results show that under the same space quota, SAQ consistently achieves lower approximation error than the baselines. In particular, to achieve the same accuracy as 8-bit RaBitQ, SAQ only needs 5-6 bits for each dimension, which is also shown  in Figure~\ref{fig:error-tiny}. When both RaBitQ and SAQ use 8 bits for each dimension, the approximation errors of SAQ are usually below 50\% of E-RaBitQ. Moreover, SAQ can accelerate the vector quantization process of RaBitQ by \revised{up to 80x}. We also show that SAQ's accuracy gains translate to  higher query throughput for ANNS,  and that our two key designs, i.e., code adjustment and dimension segmentation, are effective.

\begin{figure}[!t]
  \centering  \includegraphics[width=0.6\linewidth]{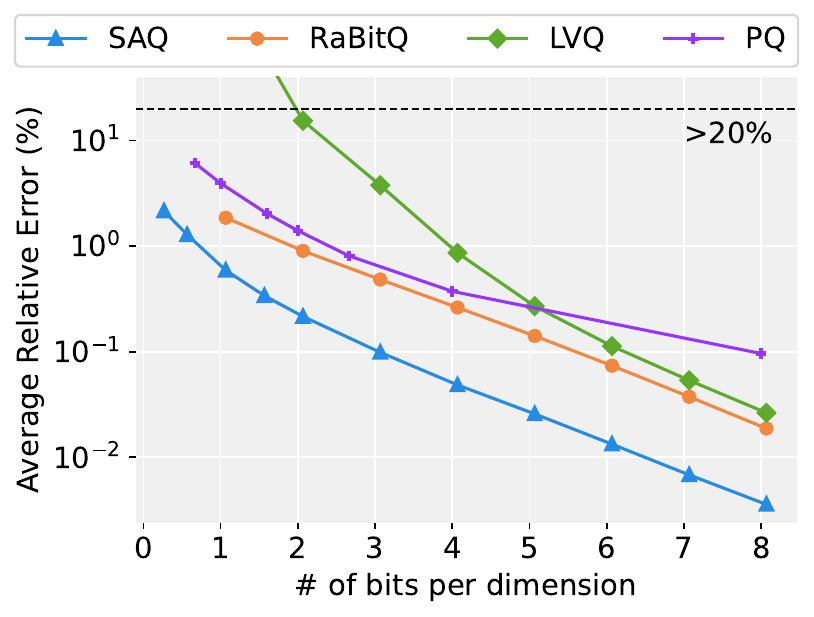}
  \caption{\revised{The vector approximation error of SAQ and representative baselines for the GIST dataset. Note that RaBitQ refers to  extended RaBitQ (same for other experiments).} }
  \label{fig:error-tiny}
\end{figure}

Our key contributions are as follows.

\squishlist
  \item We identify the critical limitations in state-of-the-art vector quantization methods, i.e., prolonged quantization time and contradictory design paradigms.

  \item We propose SAQ, which pushes the limits of current vector quantization methods in both accuracy and efficiency. SAQ comes with two key designs: 
  code adjustment for linear-time quantization  and dimension segmentation for optimal bit allocation.

  \item We design a set of novel optimizations for SAQ that includes the multi-stage distance estimator, distance bounds, dynamic programming for bit allocation, and SIMD implementations.

  \item We conduct comprehensive  experimental evaluation to demonstrate SAQ's superiority in accuracy (up to 5.4× improvement), quantization speed (\revised{up to 80× faster}), and ANNS performance (up to 12.5x higher query throughput at 95\% recall).
\squishend

\section{Preliminaries}

In this part, we introduce the basics of RaBitQ and extended RaBitQ to facilitate the subsequent discussions.

\subsection{Locally-adaptive Vector Quantization}
\revised{
Locally-adaptive Vector Quantization (LVQ)~\cite{LVQ}, a recently proposed vector quantization technique that efficiently compresses high-dimensional vectors while preserving the accuracy of similarity computations. }

\revised{
For each vector $\mathbf{x} = [x_1, \ldots, x_d]$, LVQ first mean-centers the vector by subtracting its mean $\mathbf{\mu}$ resulting in $\mathbf{x}' = \mathbf{x} - \mathbf{\mu}$, where $\mathbf{\mu}=[\mu_1, ..., \mu_d]$ is the mean of all vectors in a dataset. It then computes the minimum and maximum values of the mean-centered vector, $\ell = \min_j x'_j$ and $u = \max_j x'_j$, and divides the range $[\ell, u]$ into $2^B$ equal intervals, where $B$ is the number of bits per dimension. Each coordinate $x'_j$ is quantized to the nearest interval center:
\begin{equation}
    Q(x'_j; B, \ell, u) = \ell + \Delta  \left\lfloor \frac{x'_j - \ell}{\Delta} + \frac{1}{2}\right\rfloor, \quad \text{where } \Delta = \frac{u - \ell}{2^B - 1}.
\end{equation}
The quantized codes for all dimensions, together with $\ell$ and $u$, are stored for each vector. LVQ is simple and efficient, and by adapting the quantization range to each vector, it can achieve better accuracy than global quantization under the same bit budget. However, its accuracy may degrade under high compression rates due to the limited dynamic range per vector.
}

\subsection{RaBitQ}

RaBitQ~\cite{rbq} compresses a \( D \)-dimensional vector $x$ into a \( D \)-bit string and two floating point numbers. It provides an unbiased estimator for squared Euclidean distance and supports efficient distance computation with bit operations.

\stitle{Quantization procedure} Given a data vector \( \mathbf{o}_r \) and a query vector \( \mathbf{q}_r \), RaBitQ first normalizes them based on a reference vector \( \mathbf{c} \) (e.g., the centroid of the dataset or a cluster). The normalized vectors are expressed as
\begin{equation}
\mathbf{o} := \frac{\mathbf{o}_r - \mathbf{c}}{\|\mathbf{o}_r - \mathbf{c}\|}, \quad \mathbf{q} := \frac{\mathbf{q}_r - \mathbf{c}}{\|\mathbf{q}_r - \mathbf{c}\|}.
\end{equation}
RaBitQ estimates the inner product between the normalized vectors (i.e.,  \( \langle \mathbf{o}, \mathbf{q} \rangle \)) and computes the Euclidean distance between the original vectors (i.e., \( \mathbf{o}_r \) and \( \mathbf{q}_r \)) as:
\begin{equation}
    \begin{aligned}
    & \|\mathbf{o}_r - \mathbf{q}_r\|^2 = \|(\mathbf{o}_r  - \mathbf{c}) - (\mathbf{q}_r - \mathbf{c})\|^2
    \\
    = & \|\mathbf{o}_r - \mathbf{c}\|^2 + \|\mathbf{q}_r - \mathbf{c}\|^2 - 2 \cdot \|\mathbf{o}_r - \mathbf{c}\| \cdot \|\mathbf{q}_r - \mathbf{c}\| \cdot \langle \mathbf{q}, \mathbf{o} \rangle,
    \end{aligned}
\end{equation}
where \( \|\mathbf{o}_r - \mathbf{c}\| \) and \( \|\mathbf{q}_r - \mathbf{c}\| \) are the distances of the data and query vectors to the reference vector \( \mathbf{c} \), which can be precomputed prior to distance computation and reused. As such, RaBitQ quantizes the normalized data vector $\mathbf{o}$ and focuses on estimating  \( \langle \mathbf{q}, \mathbf{o} \rangle \).

To conduct quantization, RaBitQ first generates a random orthonormal matrix \( P \) and projects the normalized data vector \( \mathbf{o} \) by $P$, \footnote{RaBitQ actually projects \( \mathbf{o} \)  by $P^{-1}$ but this is equivalent since $P^{-1}$ is also a random orthonormal matrix.} that is, $\mathbf{o}_p := P \cdot \mathbf{o}$. Since \( P \) is orthonormal, it preserves inner product, i.e., $\langle \mathbf{q}_p, \mathbf{o}_p \rangle = \langle \mathbf{q}, \mathbf{o} \rangle $. As such, we can estimate compress $\mathbf{o}_p$ and estimate $\langle \mathbf{q}_p, \mathbf{o}_p \rangle$. Considering that \( \mathbf{o}_p \) has unit norm, RaBitQ uses the following codebook \(C\)
\begin{equation}
C := \left\{ +\frac{1}{\sqrt{D}}, -\frac{1}{\sqrt{D}} \right\}^D,
\end{equation}
where each codeword is also a unit vector, and \(\mathbf{o}_p \) is quantized to its nearest codeword in  \(C\). This essentially becomes selecting between -1 and +1 according to the sign for each dimension of \(\mathbf{o}_p \). This is also natural since after random orthonormal projection, the dimensions of $\mathbf{o}_p$ follow the same distribution, and the same number of bits should be used to quantize each dimension. The quantized vector is denoted as $\bar{\mathbf{o}}_p$.

\stitle{Distance estimator} To estimate distance, RaBitQ first rotates the normalized query vector with the random orthonormal matrix \( P \), that is, \(\mathbf{q}_p := P \cdot \frac{\mathbf{q}_r-\mathbf{c}}{\|\mathbf{q}_r-\mathbf{c}\|}\). Then, it estimates \( \langle \mathbf{o}_p, \mathbf{q}_p \rangle \) as
\begin{equation}\label{equ:estimator}
\langle \mathbf{o}_p, \mathbf{q}_p \rangle =\frac{\langle \bar{\mathbf{o}}_p, \mathbf{q}_p \rangle}{\langle \bar{\mathbf{o}}_p, \mathbf{o}_p \rangle},
\end{equation}
where $\bar{\mathbf{o}}_p$ is the quantized data vector, $\langle \bar{\mathbf{o}}_p, \mathbf{q}_p \rangle$ is computed at query time, and $\langle \bar{\mathbf{o}}_p, \mathbf{o}_p \rangle$ is computed offline and stored for each vector. It has been shown that the estimator is unbiased  and have tight error bound. We restate the lemmas as follows.

\begin{lemma} [Estimator and Error Bound]
  \label{lemma:rbq-estimator}
  The estimator of inner product is unbiased because :
  \begin{equation}
      \mathbb{E}\left( \langle \mathbf{o}_p, \mathbf{q}_p \rangle \right) = \frac{\langle \bar{\mathbf{o}}_p, \mathbf{q}_p \rangle}{\langle \bar{\mathbf{o}}_p, \mathbf{o}_p \rangle}.
  \end{equation}
  With a probability of at least $1-\exp(-c_0\epsilon_0^2)$, the error bound of the estimator satisfies
  \begin{small}
  \begin{equation}
    \mathbb{P} \left\{ \left| \frac{\langle \bar{\mathbf{o}}_p, \mathbf{q}_p' \rangle}{\langle \bar{\mathbf{o}}_p,   \mathbf{o}_p' \rangle} - \langle \mathbf{o}_p', \mathbf{q}_p' \rangle \right| > \sqrt{\frac{1 - \langle \bar{\mathbf{o}}_p, \mathbf{o}_p' \rangle^2}{\langle \bar{\mathbf{o}}_p, \mathbf{o}_p' \rangle^2}} \cdot \frac{\epsilon_0}{\sqrt{D - 1}} \right\} \leq 2e^{-c_0 \epsilon_0^2}
    \label{eq:rbq-error-bound}
  \end{equation}
  \end{small}
  where $c_0$ is a constant and $\epsilon_0$ is a parameter that controls the probability of failure of the bound.
\end{lemma}
It is also shown that $\langle \bar{\mathbf{o}}_p, \mathbf{o}_p \rangle$ is highly concentrated around $0.8$, and the estimation error is smaller than $O(1/\sqrt{D})$ with high probability~\cite{8104097}.

\subsection{Extended RaBitQ}\label{subsec:e-rabitq}

\begin{figure}[!t]
  \centering
  \includegraphics[width=0.8\linewidth]{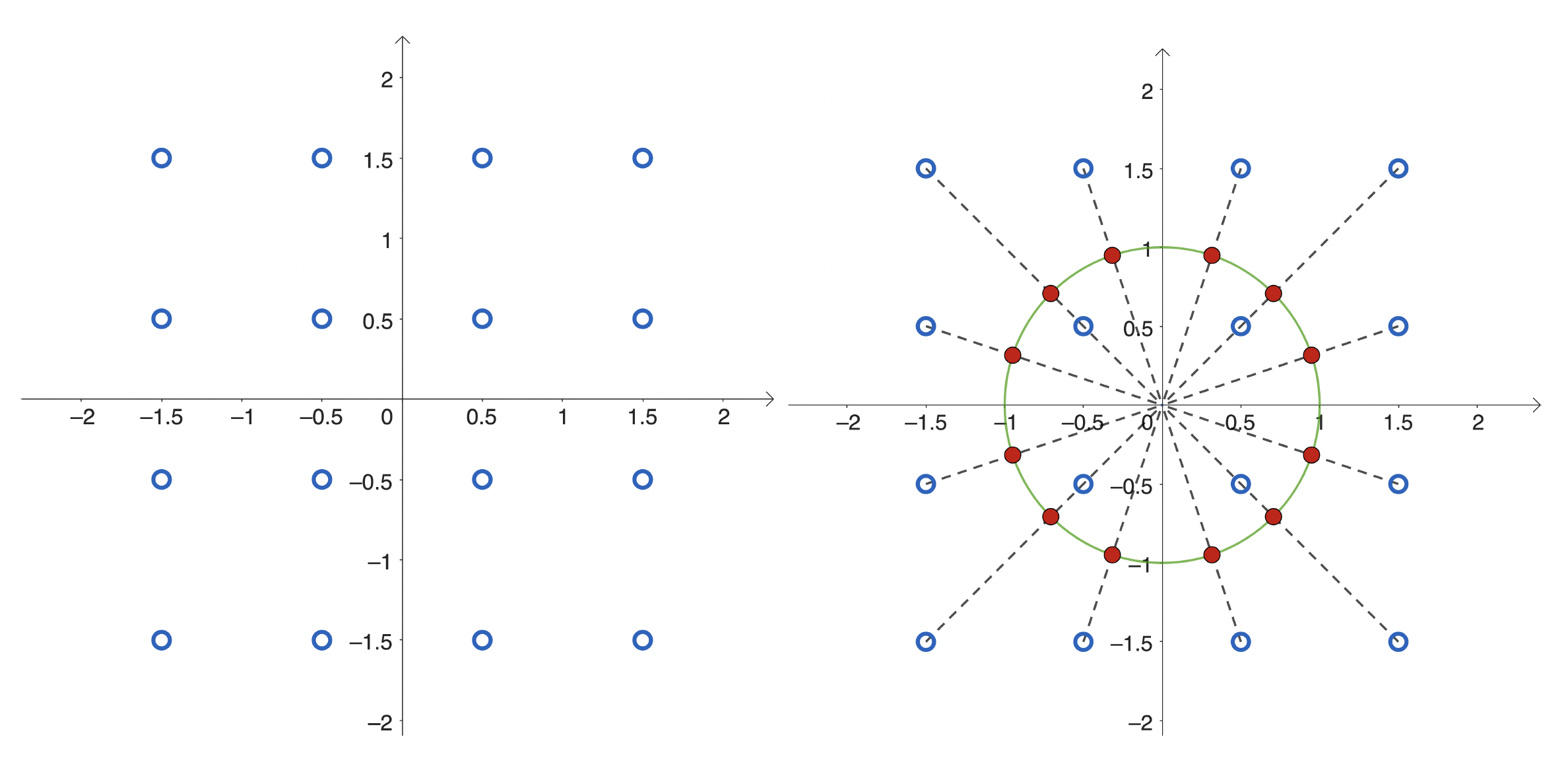}
  \caption{The codebook structure of extended RaBitQ with dimension $D\!=\!2$ and $B\!=\!2$ bits for each dimension. Red points are the final codewords, Figure reproduced from \cite{extrbq}.}
  \label{fig:rbq-codebook}
\end{figure}

RaBitQ only allows to use 1-bit for each vector dimension, which limits accuracy. To improve accuracy, it pads $\mathbf{o}$ with zeros to a dimension $D'>D$ and uses $D$ bits for quantization. However, this is shown to be sub-optimal, and extended RaBitQ~\cite{extrbq} (denoted as E-RaBitQ) is proposed to support \(B\)-bit quantization for each vector dimension, where \(B\) is a positive integer. In particular, E-RaBitQ quantizes a \(D\)-dimensional vector \( \mathbf{o}_r \) into a \((D*B)\)-bit string and two floating point numbers.

The normalization and projection of E-RaBitQ are the same as RaBitQ but E-RaBitQ uses the following codebook $\mathcal{G}_{r}$ to quantize \(\mathbf{o}_p\)
\begin{small}
\begin{equation}
  \label{eq:rbq-codebook}
  \begin{aligned}
  \mathcal{G} := & \left\{ -\frac{2^B -1}{2} +u \, \middle| \, u = 0, 1, 2, 3, ..., 2^B - 1 \right\}^D \\
  \mathcal{G}_{r} := & \left\{ \frac{\mathbf{y}}{\left\| \mathbf{y} \right\|}  \middle|  \mathbf{y}\in G \right\}^D.
  \end{aligned}
\end{equation}
\end{small}
As shown in Figure~\ref{fig:rbq-codebook}, the raw codewords form a regular grid in $D$-dimensional space, while the actual codewords are scaled to the unit sphere to have unit norm. This is because \(\mathbf{o}_p\) also has unit norm.

A data vector \(\mathbf{o}_p\) finds the nearest codeword in $\mathcal{G}_{r}$ as its quantized vector and stores the corresponding quantization code \(\bar{\mathbf{o}}_b \in \left\{ 0,1,2,3,...,2^B-1 \right\}^D \) for the codeword. However, it is difficult to find the neatest codeword in $\mathcal{G}_{r}$. E-RaBitQ propose a pruned enumeration algorithm to search the codeword with a time complexity of $O(2^B\cdot D\log{D})$. As such, the quantization time can be long with using a few bits (e.g. $B\ge 8$) for each dimension for high accuracy.

The distance estimator of E-RaBitQ is the same as RaBitQ, and its error bound is shown empirically as follows.
\begin{remark} [Error Bound]
  Let $\epsilon$ be the absolute error in estimating the inner product of unit vectors. With >99.9\% probability, we have $\epsilon < 2^{-B} \cdot c_{\epsilon} / \sqrt{D}$ where $c_{\epsilon}=5.75$
\end{remark}
This is asymptotically optimal as the error scales with $2^{-B}$.

It has been shown that when $B>1$, taking the most significant bit for each dimension can get the quantized vector of the original RaBitQ. As such, E-RaBitQ proposes a progressive strategy to accelerate vector search, which first uses the most significant bits of $\bar{\mathbf{o}}_p$ to compute bounds for the distance, and full-precession computation is used only when $\mathbf{o}$ is likely to have smaller distance than the current top-$k$ neighbors. However, taking $b<B$ but $b>1$ bits from  $\bar{\mathbf{o}}_p$ does not necessary form a valid quantized vector. In the subsequent paper, we denote E-RaBitQ as RaBitQ for simplicity.

\section{Code Adjustment}

As discussed in Section~\ref{subsec:e-rabitq}, RaBitQ has a complexity of \(O(2^{B}\cdot D \log{D})\) for encoding a $D$-dimension vector with $B\cdot D$ bits. Therefore, vector quantization can be time consuming. For example, for a dataset with a billion vectors, $D=3072$ and $B=9$, RaBitQ  takes more than 3,600 CPU hours. As we mentioned in Section~\ref{sec:intro}, SAQ segments dimensions and uses more bits for segments with large magnitudes. 
If we directly use RaBitQ for each segment, it could lead to a longer quantization time since $B>10$ bits may be used for each dimension of the leading segments and RaBitQ's indexing time grows exponentially with \(B\).

To reduce quantization time, we propose \textit{code adjustment quantization} (\textit{CAQ}), which reduces the quantization complexity drastically from  \(O(2^{B}\cdot D \log{D})\) to \(O(D)\), while maintaining the same empirical error and efficiency for distance estimation as RaBitQ. Another special feature of CAQ is, using $B$ bits for each dimension, taking first $b$ bits ($1<b<B$) for each dimension from the quantized code $\bar{\mathbf{c}}_a$ still forms a valid quantized vector. This provides more opportunities for progressive distance computation, which is not supported in RaBitQ. Table~\ref{tab:notation} lists the frequently used notations.

\begin{table}
  \caption{Frequently used notations in the paper}
  \label{tab:notation}
  \begin{tabular}{ll}
    \toprule
    \textbf{Notation} & \textbf{Definition} \\
    \midrule
    \(\mathbf{o}_r\), \(\mathbf{q}_r\) & the original data and query vectors  \\
    \(\mathbf{o}\), \(\mathbf{q}\) & \revised{the rotated and deducted data and query vectors}  \\
    \(\bar{\mathbf{o}}\) & the quantized vector of \(\mathbf{o}\) \\
    \(\bar{\mathbf{o}}_a\) & the adjusted quantized vector of \(\mathbf{o}\) \\
    \(\bar{\mathbf{c}}, \bar{\mathbf{c}}_a\) & the quantization code of \(\bar{\mathbf{o}}\) and $\bar{\mathbf{o}}_a$ \\
    $[C]$ & an integer set with $\{0, 1, 2, \cdots, C\}$ \\
    \bottomrule
  \end{tabular}
\end{table}

\stitle{Insights} RaBitQ has high quantization complexity because it requires the quantized vector (and thus vector codewords) $\bar{\mathbf{o}}_p$ to have unit norm. This makes the dimensions of $\bar{\mathbf{o}}_p$ dependent as their squared values must sum to 1 and prevents handling each dimension of $\mathbf{o}_p$
 independently.
Such a design is reasonable at first glance since the norm $\|\mathbf{o}_r - \mathbf{c}\|$ is stored explicitly, and $\bar{\mathbf{o}}_p$ quantizes the direction vector $\mathbf{o}_p$, which should have unit norm. However, by inspecting the estimator in Eq~\eqref{equ:estimator}, we observe that the norm of $\bar{\mathbf{o}}_p$ does not affect estimation. This is, $\bar{\mathbf{o}}_p$ appears in both the numerator and denominator, and thus scaling $\bar{\mathbf{o}}_p$ will not change the estimated value. Instead, $\bar{\mathbf{o}}_p$ should only be required to align with $\mathbf{o}_p$ in direction. Therefore, we remove the unit norm constraint for $\bar{\mathbf{o}}_p$ and consider each dimension of $\mathbf{o}_p$ individually. Intuitively, CAQ works like coordinate descent in optimization, i.e., it starts with a raw quantized vector and then adjusts its dimensions to better align with  $\mathbf{o}_p$ in direction.

\subsection{Quantization Procedure}

Like RaBitQ, CAQ deducts a reference vector $\mathbf{c}$ from the data and query vectors and rotates them with a random orthonormal matrix \(P\). Let \(\mathbf{o}_r\) and \(\mathbf{q}_r\) be the raw data and the query vectors, and the rotated vectors are:
\begin{equation}
  \mathbf{o} = P \cdot (\mathbf{o}_r - \mathbf{c}), \quad \mathbf{q} = P \cdot (\mathbf{q}_r - \mathbf{c}).
  \label{eq:rotate}
\end{equation}
The distance between the original data vector \(\mathbf{o}_r\) and query vector \(\mathbf{q}_r\) is equal to the distance between the rotated data vector \(\mathbf{o}\) and query vector \(\mathbf{q}\) since \(P\) is orthonormal. For simplicity, we also call \(\mathbf{o}\) and \(\mathbf{q}\) data and query vectors.

\revised{ CAQ initializes the quantized vector $\bar{\mathbf{o}}$ similar with LVQ~\cite{LVQ}. }
In particular, for data vector $\mathbf{o}$, let $v_{\max} = \max_{i\in[D]} |\mathbf{o}[i]|$, i.e., the maximum magnitude in $\mathbf{o}$. We divide the range of $[-v_{\max}, v_{\max} ]$ into $2^B$ uniform intervals, each with length $\Delta = (2\cdot v_{\max})/ 2^B$. The midpoint of the $x^{th}$ interval is $-v_{max} + \Delta(x+0.5)$, which is used to quantify the dimensions of $\mathbf{o}$ whose value falls in its corresponding interval. These midpoints form a $D$-dimensional uniform grid, and for each dimension $i\in [D]$, we have
\begin{equation}
  \begin{aligned}
    \label{eq:find_nearest_code}
    \bar{\mathbf{c}}[i] & := \left\lfloor \frac{\mathbf{o}[i] + v_{\max}}{\Delta} \right\rfloor
  \end{aligned}
\end{equation}
\noindent as the quantization code. Data vector $\mathbf{o}$ is quantized as
\begin{equation}
    \bar{\mathbf{o}} := \Delta (\bar{\mathbf{c}} +0.5) - v_{\max}\cdot\mathbf{1}_D.
\end{equation}
LVQ stops here and uses $\bar{\mathbf{o}}$ to estimate  distance. \revised{ However, according to RaBitQ's analysis \cite{extrbq}, to approximate the data vector, the quantization vector should align with the data vector $\mathbf{o}$ in direction.} As such, quantization should find the vector $\mathbf{x}$ with the largest cosine similarity to $\mathbf{o}$, which is defined as
\begin{equation}
\label{eq:adj-obj}
\mathcal{L}(\mathbf{x}, \mathbf{o}) = \frac{\mathbf{x} \cdot \mathbf{o}}{\left\|\mathbf{x}\right\|_2 \cdot \left\|\mathbf{o}\right\|_2}.
\end{equation}

\revised{
LVQ only produce the quantization vector $\bar{o}$ that is the nearest to the data vector $\mathbf{o}$ among all possible quantization vectors and it may not optimize the objection function in Eq.~\eqref{eq:adj-obj}. However, as we have observed, the cosine similarity can be significantly improved by adjusting and refining $\bar{\mathbf{o}}$ to better align with the data vector $\mathbf{o}$. We propose an efficient code adjustment algorithm shown in Algorithm~\ref{alg:find-adjusted-quantization}.
Lines 7-11 try to adjust a dimension of the quantized vector by a step of $\Delta$ and see if the direction is better aligned. The adjustment iterates over all dimensions and for a limited number of rounds (Lines 5-6). }

\revised{
Figure~\ref{fig:step34} illustrates how the code adjustment algorithm adjusts the initial code produced from LVQ to improve cosine similarity. }

\begin{algorithm}[!t]
  \caption{Adjustment of the LVQ Quantized Vector}
  \label{alg:find-adjusted-quantization}
  \textbf{Input}: round limit $r$, data vector $\mathbf{o}$, start point $\bar{\mathbf{o}}$, value $v_{\max}$ \\
  \textbf{Output}: the final quantized vector $\bar{\mathbf{o}}_a$ for data vector $\mathbf{o}$

  \BlankLine
  Step size $\Delta \leftarrow (2\cdot v_{\max})/2^B$ \;
  Initialize $\mathbf{x} \leftarrow \bar{\mathbf{o}}$ \;

  \For{$round \leftarrow 1$ \KwTo $r$}{
    \For{$i \leftarrow 1$ \KwTo $D$}{
      \For{$\delta \in \{\Delta, -\Delta \}$}{
        $\mathbf{x}' \leftarrow \mathbf{x}$ \;
        $\mathbf{x}'[i] \leftarrow \mathbf{x}'[i] + \delta$ \;
        \If{$\mathbf{x}'[i]  \in [-v_{\max}, v_{\max}] \land \mathcal{L}(\mathbf{x}',\mathbf{o}) > \mathcal{L}(\mathbf{x},\mathbf{o})$}{
          $\mathbf{x} \leftarrow \mathbf{x}'$ \;
        }
      }
    }
  }
  $\bar{\mathbf{o}}_a \leftarrow \mathbf{x}$

  \Return{$\bar{\mathbf{o}}_a$}
\end{algorithm}

Since Algorithm~\ref{alg:find-adjusted-quantization} only changes one vector dimension for each adjustment, we do not need to recompute $\mathcal{L}(\mathbf{x}',\mathbf{o})$ by enumerating all dimensions. Instead, we only need to recompute the contribution of the current dimension to $\mathcal{L}(\mathbf{x}',\mathbf{o})$. Thus, the time complexity of each adjustment is $O(1)$.
The overall time complexity of Algorithm~\ref{alg:find-adjusted-quantization} is $O(r\cdot D)$, which is in the same $O(D)$ order for computing $\bar{\mathbf{o}}$. \revised{We evaluate the quantization accuracy with different number of adjustment iteration $r$ in Section ~\ref{sec:micro_results}.
In practice, we recommend setting $r\in[4, 8]$, which is sufficient to obtain a quantized vector with high quality.}

After adjustment, we obtain the approximate quantization vector $\bar{\mathbf{o}}_a$ and compute the final quantization code $\bar{\mathbf{c}}_a$ using Eq.~\eqref{eq:find_nearest_code}. The final quantization code is a $D$-dimensional vector whose coordinates are $B$-bit unsigned integers so that we can store the code with a $(B*D)$-bit string. Like RaBitQ, CAQ uses two additonal float numbers for each vector to store the norm and $\langle \bar{\mathbf{o}}_a, \mathbf{o} \rangle$.

\begin{figure}[t]
  \centering
  \includegraphics[width=0.8\linewidth]{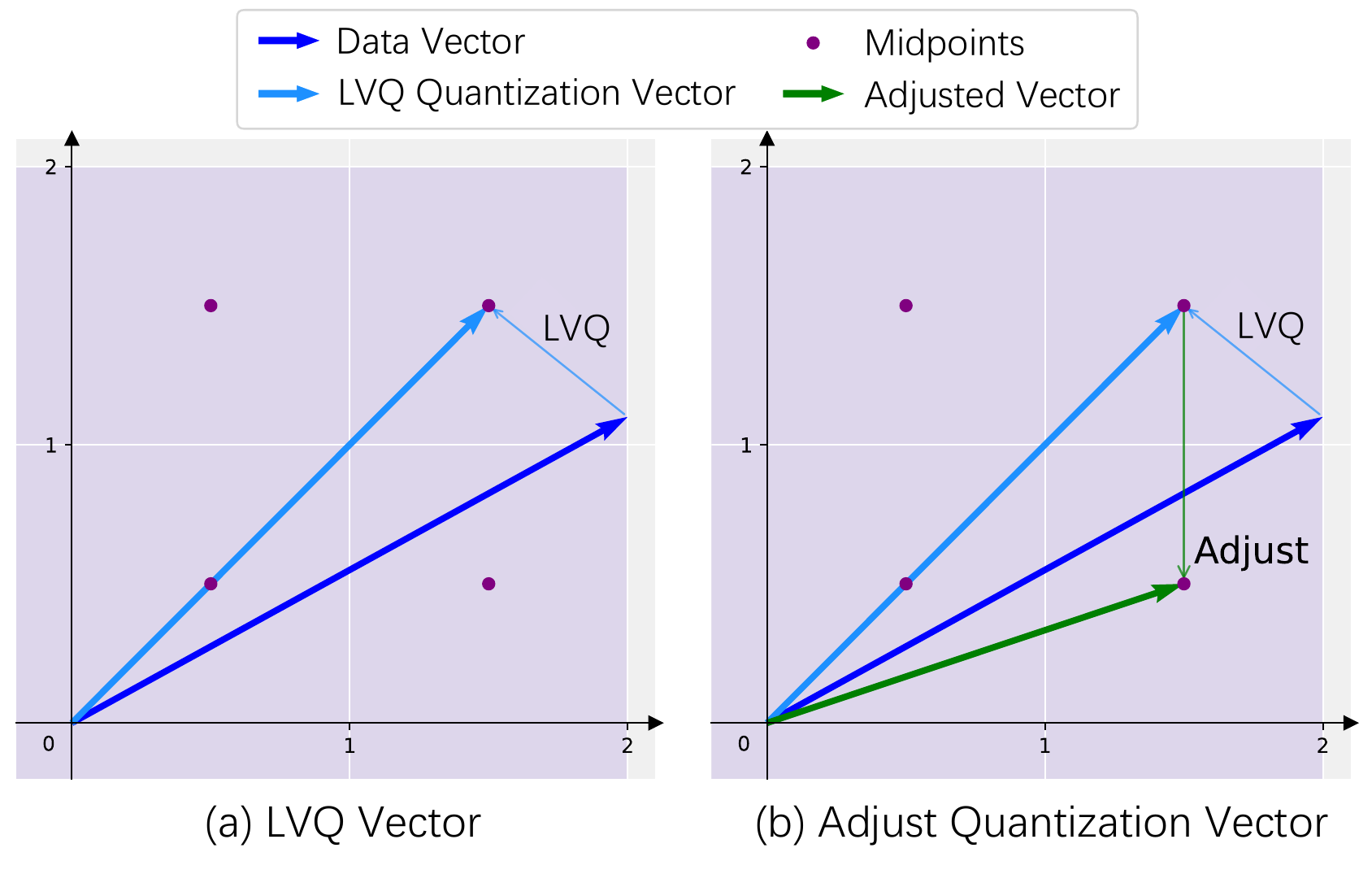}
  \caption{\revised{The procedure of CAQ for quantization, which first starts with LVQ and then adjusts the quantized vector to align with the data vector in direction.}}
  \label{fig:step34}
\end{figure}

\subsection{Distance Estimation} \label{sec:DistEst}

We adopt the same distance estimator as RaBitQ, i.e., approximating $\langle \mathbf{o}, \mathbf{q}\rangle$ as $\langle \bar{\mathbf{o}}_a, \mathbf{q} \rangle/\langle \bar{\mathbf{o}}_a, \mathbf{o} \rangle \cdot |\mathbf{o}|^2$. As discussed earlier, $\langle \bar{\mathbf{o}}_a, \mathbf{o} \rangle$ is precomputed and stored. In the query phase, we can compute $\langle \bar{\mathbf{o}}_a, \mathbf{q} \rangle$ using the integer quantization code $\bar{\mathbf{c}}_a$ without decompression as follows

\begin{equation}
  \begin{aligned}
    \langle \bar{\mathbf{o}}_a, \mathbf{q} \rangle = &  \sum_{i=1}^D \bar{\mathbf{o}}_a[i] \cdot \mathbf{q}[i] \\
    = & \sum_{i=1}^D \left( \Delta (\bar{\mathbf{c}}_a[i] + 0.5) - v_{\max} \right) \cdot \mathbf{q}[i] \\
    = &  \Delta \langle \bar{\mathbf{c}}_a, \mathbf{q} \rangle + q_{\text{sum}} (-v_{\max}+\Delta/2),
  \end{aligned}
    \label{eq:estimator_inner_product}
\end{equation}
where $q_{sum} = \sum_{i=1}^D \mathbf{q}[i]$, which only needs to be computed once for each query vector. Like RaBitQ, if we quantize each dimension of a query vector to integer,  Eq.~\ref{eq:estimator_inner_product} will mostly use integer computation.

\stitle{Progressive distance approximation} CAQ supports progressive distance approximation using the prefix of quantized code of each dimension with an arbitrary length.
In particular, let $\bar{\mathbf{c}}$ denote the first $b$ bits sampled from the native $B$ bit quantization code $\bar{\mathbf{c}}_a$ in each dimension, that is, $\bar{\mathbf{c}}_s = \lfloor \bar{\mathbf{c}}_a / 2^{B-b} \rfloor$. The inner product can be estimated using $\bar{\mathbf{c}}_s$ via Eq.~\ref{eq:estimator_inner_product} by replacing $\Delta$ with $\Delta'=\Delta \cdot 2^{B-b}$ and $\bar{\mathbf{c}}_a$ with $\bar{\mathbf{c}}_s$.  Our experiments in Section~\ref{sec:progressive-error} show that the prefix $\bar{\mathbf{c}}_s$ yields almost the same estimation error as a native CAQ quantized vector using $b$-bits for each dimension. The progressive distance estimator of CAQ enables vector search to conduct progressive distance refinement (e.g., first with 1-bit, then 2-bit, and finally 8-bit). This allows us to provide multiple efficiency-accuracy trade-off options from the same quantization code to satisfy various requirements (e.g., we use it in our multi-stage estimation in Section~\ref{sec:saq-estimator}). In contrast, RaBitQ only supports progressive approximation with $b=1$, which is rather restrictive in its usage.

\subsection{Analysis}

Recall that, to approximate a data vector $\mathbf{o}$, RaBitQ finds its nearest codeword in codebook $ \mathcal{G}_{r}$ that contains unit norm vectors. That is, RaBitQ solves the following optimization problem for quantization
\begin{equation}
    \operatorname{argmax}_{\mathbf{x}\in \mathcal{G}_{r}} \mathcal{L}(\mathbf{x}, \mathbf{o}).
\end{equation}
CAQ finds the codeword that best aligns with $\mathbf{o}$ in direction in a codebook $\mathcal{D}_\text{CAQ}$. That is, CAQ solves the following optimization problem:
\begin{equation}\label{equ:p-CAQ}
    \operatorname{argmax}_{\mathbf{x}\in \mathcal{D}_\text{CAQ}}  \mathcal{L}(\mathbf{x}, \mathbf{o})
\end{equation}
\noindent among $\mathcal{D}_\text{CAQ}$. We show in Lemma~\ref{lem:caq-error-bound} that the codebooks of RaBitQ  and CAQ are essentially the same.

\begin{lemma}
  \label{lem:caq-error-bound}
  $\mathcal{D}_\text{CAQ}$ of CAQ is equivalent to $\mathcal{G}_r$ in RaBitQ.
\end{lemma}
\begin{proof}
From Lines 6-11 of Algorithm \ref{alg:find-adjusted-quantization}, the dimensions of $\mathbf{x}$ take values from $\left\{-v_{\max}+\Delta \cdot (i + 0.5) \mid i \in [2^{B} - 1]\right\}$. Thus, we have
$$
    \mathcal{D}_\text{CAQ} = \left\{-v_{\max}+\Delta \cdot (i + 0.5) \mid i \in [2^{B} - 1]\right\}^{D}.
$$
\noindent With $\Delta = 2 \cdot v_{\max} / 2^B$, dividing the vectors in $\mathcal{D}_\text{CAQ}$ by $v_{\max}$ gives $\left\{ -(2^B -1)/{2} + i \, \middle| \, i \in [2^{B} - 1] \right\}^D$, which matches the unnormalized  codebook $\mathcal{G}$ of RaBitQ in Eq.~\ref{eq:rbq-codebook}. The normalization only changes the norm of the codewords but not the direction, and we have discussed that the norms of the codewords do not affect the estimated inner product.
\end{proof}
Therefore, if CAQ can solve its optimization problem in Eq~\eqref{equ:p-CAQ}, it achieves the unbiased estimation and error bound as RaBitQ. Although the coordinate descent style optimization of Algorithm~\ref{alg:find-adjusted-quantization} does not necessarily converge to the optimal, empirically, we observe that CAQ achieves identical estimation errors as RaBitQ.

\section{Dimension Segmentation}
\label{sec:SAQ}

In this section, we present \textit{Segmented CAQ} (\textit{SAQ}). SAQ combines the benefits of both \textit{dimension balancing} and \textit{dimension reduction} to push the performance limits of vector quantization. In Section \ref{sec:saq-motivation}, we first introduce the motivation of SAQ and formulate the problem of finding a quantization plan. Then we present a dynamic programming algorithm to find the optimal quantization plan in Section \ref{sec:saq-dp}.
For query processing, in Section~\ref{sec:saq-estimator}, we show how to use the quantization plan to estimate distance. We also present a multi-stage estimator that facilitates candidate pruning based on the property of the quantization plan.


\subsection{Motivation and Problem Formulation}
\label{sec:saq-motivation}

CAQ requires the random orthonormal projection because it treats each vector dimension equivalently by quantizing them with the same number of bits. If some dimensions have much larger variances than the others, their accuracy will hurt. In an extreme example, consider a dataset of two-dimensional vectors whose values for the first dimension follow the normal distribution and the values for the second dimension are the same fixed value. The bits assigned for the second dimension will be wasted as we can store one copy of the specific value for all data vectors. By a random projection, the variance is scattered across all dimensions to ensure that all bits are fully utilized to boost accuracy. We refer to this technique as \textbf{ dimension balancing}.

In the reverse direction, if we concentrate the data vectors' variances into certain dimensions, we can remove the dimensions with negligible variances. PCA is widely used for this purpose. A recent study~\cite{wangweipca} indicates that the high-dimensional vector coordinates, once rotated by a PCA matrix, exhibit a long-tailed variance distribution, which implies that only a few dimensions hold most of the variance (Figure \ref{fig:pca-longtail}). This characteristic was used to create an approximate distance estimator that uses the leading dimensions. These techniques are commonly known as \textbf{dimension reduction}.
Dimension reduction is less general, as it is highly dependent on the data distribution. When PCA fails to polarize the variance, a fixed dimension reduction ratio will result in poor accuracy.

The common idea behind these two methods indicates that dimensions with higher variance should be allocated more bits and those with lower variance fewer bits. Building on this observation, we introduce \textit{Segmented CAQ} (\textit{SAQ}), which uses varying compression ratios based on dimension variance to improve accuracy. Unlike a common parameter $B$ for all dimensions as in CAQ,  SAQ takes a parameter $Q_{quota}$ to represent the total bit quota to quantize the entire vector. Given a dataset, we first learn a PCA rotation matrix and rotate all data vectors with the PCA matrix. In this way, we polarize the variance and obtain $\sigma_1 \geq \sigma_2 \geq \cdots \geq \sigma_D$, where $\sigma_i$ is the variance of the values taken by the $i^{\text{th}}$ dimensions of the rotated vectors. SAQ  then performs quantization following a quantization plan, and an example is illustrated in Figure~\ref{fig:segment-fig}.


\begin{figure}[!]
  \subfigure[DEEP]{
    \begin{minipage}[t]{0.45\linewidth}
      \centering
      \includegraphics[width=\linewidth]{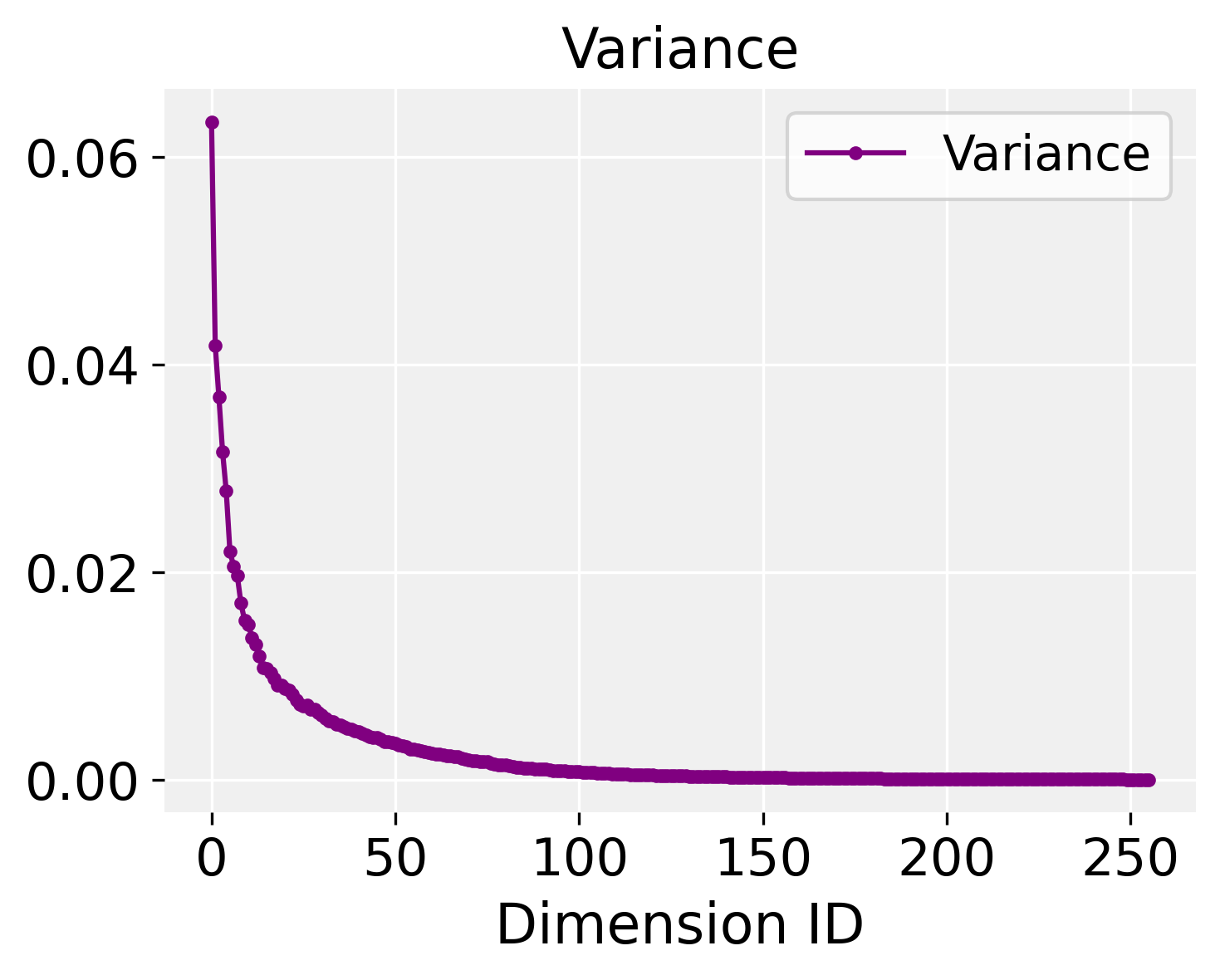}
    \end{minipage}
  }
  \subfigure[OpenAI-1536]{
    \begin{minipage}[t]{0.45\linewidth}
      \centering
      \includegraphics[width=\linewidth]{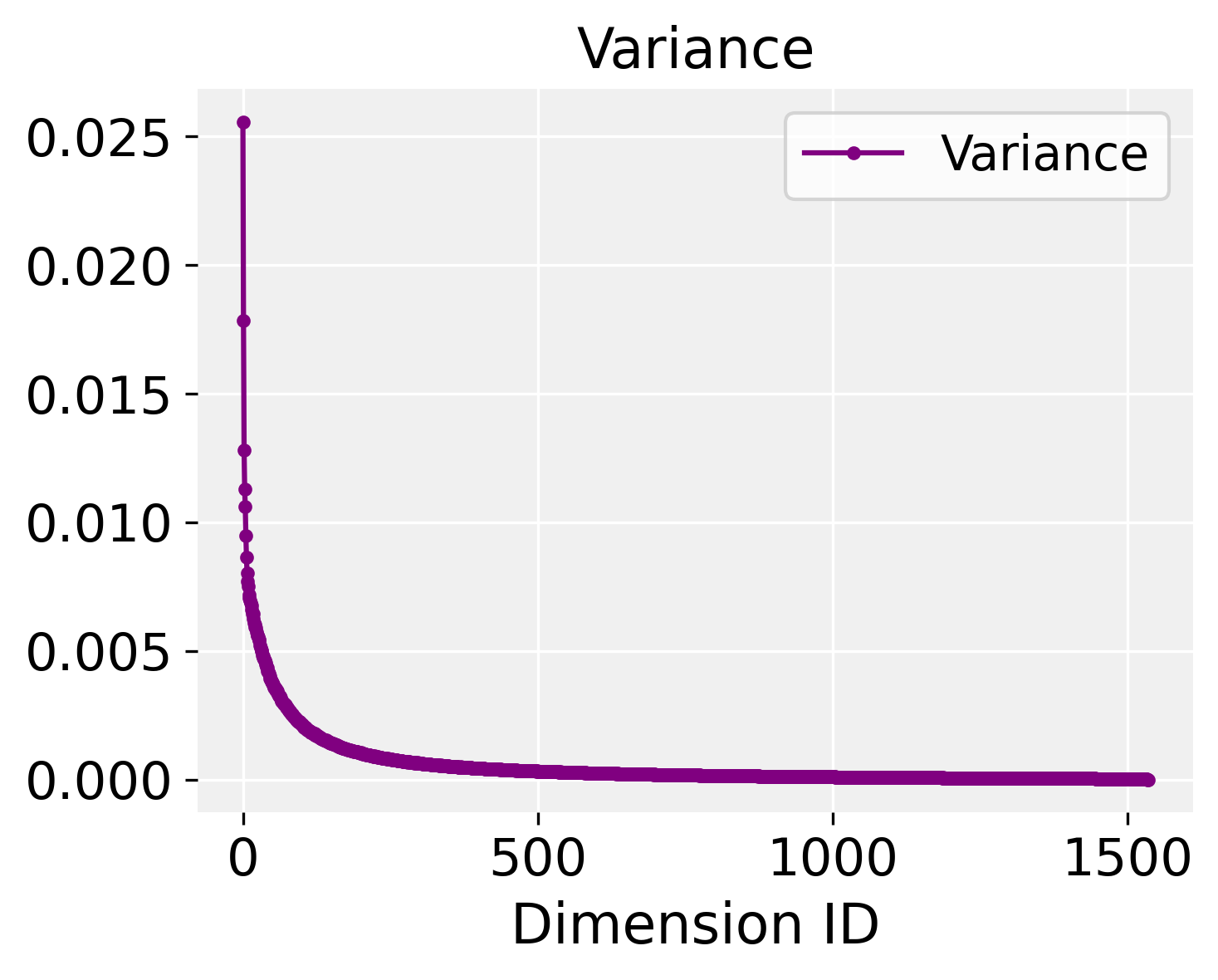}
    \end{minipage}
  }

  \caption{Variance of vector dimension after PCA projection.}
  \Description{figure description}
  \label{fig:pca-longtail}
\end{figure}

\begin{figure}
  \includegraphics[width=0.9\linewidth]{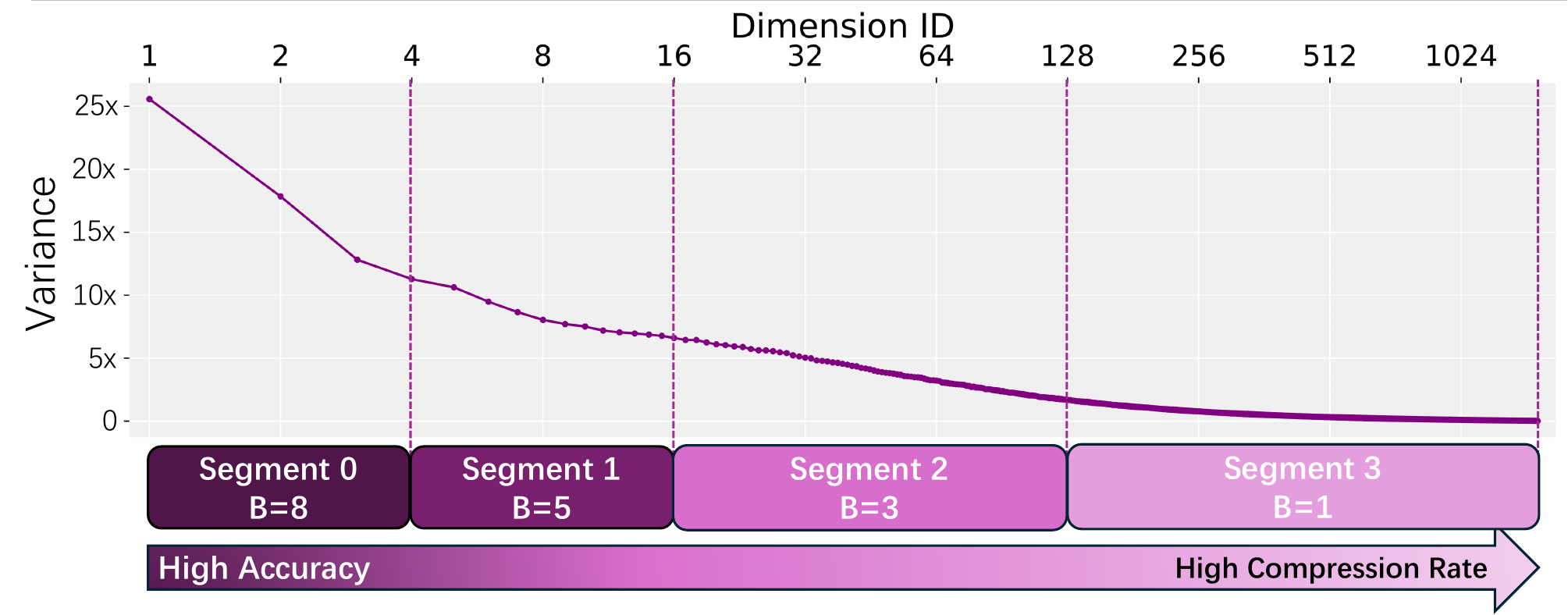}
  \caption{An illustration of dimension segmentation.}
  \Description{figure description}
  \label{fig:segment-fig}
\end{figure}

To define a quantization plan, we first introduce the concept of a dimension segment $\mathit{Seg}$, which is a set of continuous vector dimensions. For a dimension segment $\mathit{Seg}$ and a vector $x\in X$, $x[\mathit{Seg}]$ indicates the projection of $x$ onto the dimensions in $\mathit{Seg}$, and $X[\mathit{Seg}]$ represents the set of all such projected vectors.
We further define a tuple $(\mathit{Seg}, \mathit{B})$ that indicates that SAQ will quantize each dimension in $\mathit{Seg}$ using $\mathit{B}$ bits. A quantization plan is given by a set of such tuples $\mathcal{P} := \{(\mathit{Seg}_1, \mathit{B}_1), (\mathit{Seg}_2, \mathit{B}_2), ...\}$, where $\mathit{Seg}_i$ is a dimension segment and $\mathit{B}_i$ is the number of bits used to quantize the dimensions in $\mathit{Seg}_i$. The union of all these $\mathit{Seg}_i$ is $[D]$. The  total bit consumption of this plan is
\begin{equation}
  \mathcal{C}(\mathcal{P}) = \sum_{(\mathit{Seg}, \mathit{B})\in \mathcal{P}} \mathit{B}\cdot |\mathit{Seg}|.
\end{equation}
\noindent We want to find a quantization plan $\mathcal{P}$ that achieves a low estimation error with $\mathcal{C}(\mathcal{P}) \leq Q_{quota}$. The number of possible quantization plans can be prohibitively large, and it is difficult to predict the estimation error without deploying and testing the plan with real queries. To this end, we provide an efficient way to model the estimation error of a  quantization plan and a method to search for good quantization plan. Note that when quantizing each segment with CAQ, we still apply the random orthonormal projection beforehand such that the dimensions are similar in variance.

\subsection{Quantization Plan Construction}
\label{sec:saq-dp}

According to our experiments and the analysis of RaBitQ \cite{extrbq}, the relative error of quantization a segment (or vector) is proportional to $2^{-B}$, where $B$ is the number of bits used to quantify the segment. This is natural because with 1 more bit, we can reduces the quantization resolution by $1/2$. Moreover, we observe that, after PCA projection, the value distribution of a dimension almost follows a normal distribution with a mean close to $0$, as illustrated in Figure \ref{fig:value-distr}. Therefore, we propose the following expression to model the error introduced 
in the index phase:
\begin{equation}
  \operatorname{ERROR}(\mathit{Seg}, \mathit{B}) = \sum_{i\in \mathit{Seg}} \frac{\mathbb{E}\left[|\mathbf{o}_r[i]\cdot\mathbf{q}_r[i]| \right]}{2^{B+1}} = \frac{1}{ 2^{B}\pi} \sum_{i\in \mathit{Seg}}\sigma_i^2,
  \label{eq:seg-err}
\end{equation}
where $\mathbf{o}_r$ and $\mathbf{q}_r$ are the raw data vector and raw query vector, and $\sigma_i$ means the value variance of the $i$-th dimension\footnote{Here, we assume query vectors share the same distribution with data vectors. In fact, we can drop the normal distribution assumption and simply use the empirical variance of each dimension (i.e., $\sigma_i^2$). This will remove the $\pi$ in Eq~\eqref{eq:seg-err}. }.
The error introduced by a quantization plan $\mathcal{P}$ can then be modeled as
\begin{equation}
  \operatorname{ERROR}(\mathcal{P}) = \sum_{(\mathit{Seg}, \mathit{B}) \in \mathcal{P}} \operatorname{ERROR}(\mathit{Seg}, \mathit{B}).
\end{equation}

The problem then becomes finding a plan under the given bit quota $Q_{quota}$ with minimal $\operatorname{ERROR}(\mathcal{P})$. To solve this problem, we devise a dynamic programming algorithm as presented in Algorithm \ref{alg:construct-quantization-plan}.

\begin{algorithm}[t]
  \caption{Quantization Plan Search}
  \label{alg:construct-quantization-plan}
  \textbf{Input}: vector dimension $D$, total bit quota $Q_{quota}$ \\
  \textbf{Output}: Quantization plan $\mathcal{P}$.

  \BlankLine
  $p_{0, 0} \leftarrow (\emptyset, \emptyset)$

  \For{$d \leftarrow 0$ \KwTo $D-1$}{
    \For{$Q \leftarrow 0$ \KwTo $Q_{quota}$}{
      \If{$p_{d, Q}$ exists}{
        \For{$d' \leftarrow d$ \KwTo $D$}{
          $\mathit{seg} \leftarrow \{d+1, d+2,...,d'\} $ \\

          \For{$\mathit{b}'$ is valid quantization bits}{
            $Q' \leftarrow Q + b'*d'$ \\
            $p' \leftarrow  p_{d, Q} \cup (\mathit{seg}, \mathit{b}')$ \\

            \If{$p_{d', Q'}$ not exist or $\operatorname{ERROR}(p_{d', Q'}) > \operatorname{ERROR}(p')$}{
              $p_{d', Q'} \leftarrow p'$
            }
          }
        }
      }
    }
  }
  $\mathcal{P} \leftarrow \emptyset$

  \For{$Q \leftarrow 0$ \KwTo $Q_{quota}$}{
    \If{$p_{D, Q}$ exists \and $\operatorname{ERROR}(p_{D, Q}) < \operatorname{ERROR}(\mathcal{P})$}{
      $\mathcal{P} \leftarrow p_{D, Q}$
    }
  }
  \Return{$\mathcal{P}$}
\end{algorithm}


Let $p_{d, Q}$ denote the quantization plan that contains the first $d$ dimensions and uses $Q$ bits quota in total. We use dynamic programming to find the optimal quantization plan. Specifically, each time we enumerate an existing quantization plan $p_{d, Q}$ (lines 4-6) and append a new segment $\mathit{seg}$ that contains a set of dimensions $\{d+1, d+2, ..., d'\}$ and use $b'$ bits to quantize each dimension (lines 7-11). We create this new quantization plan or update it if it already exists and the new one is better (lines 12-13). After the whole process, we obtain the optimal quantization plan $\mathcal{P}$ that minimizes the total error within the total bit quota $Q_{quota}$ (lines 15-17).

The time complexity of this dynamic programming is $O(D^2\cdot Q_{quota})$, which is insignificant compared to the time complexity of quantization as it does not loop over vectors. In practice, we set the size of each segment to be a multiple of 64 to align with cache line size. Moreover, as each segment has some extra computation overhead for the distance estimator, which we will shown soon, using many small segments is not efficient. Thus, 
we choose the quantization plan that gives an error close to (e.g., $<$0.1\% of) the minimum but with the least number of segments. On all our experimented datasets, quantization plan search can finish within 1 second.

Figure \ref{fig:segment-fig} illustrates an example of our quantization plan (disregarding the 64 limits for simplicity). We allocate more bits to dimensions with high variance to attain higher precision, while applying higher compression rates to dimensions with low variance to maintain the total bit budget.
After constructing the quantization plan, we use it to split the dimensions of data vectors into segments and quantize each segment separately just as a normal vector.

\revised{
Dimension segmentation relies on a skewed eigen value distribution over the dimensions such that more bits can be used for dimension segments with larger eigen values. When the eigen value distribution is perfectly uniform, the chance for improvement varnishes. However, for these extreme cases, the quantization plan will contain only one segment that contains all dimensions and will still match the performance of CAQ.
}

\begin{figure}[t]
  \subfigure[DEEP]{
    \begin{minipage}[t]{0.46\linewidth}
      \centering
      \includegraphics[width=\linewidth]{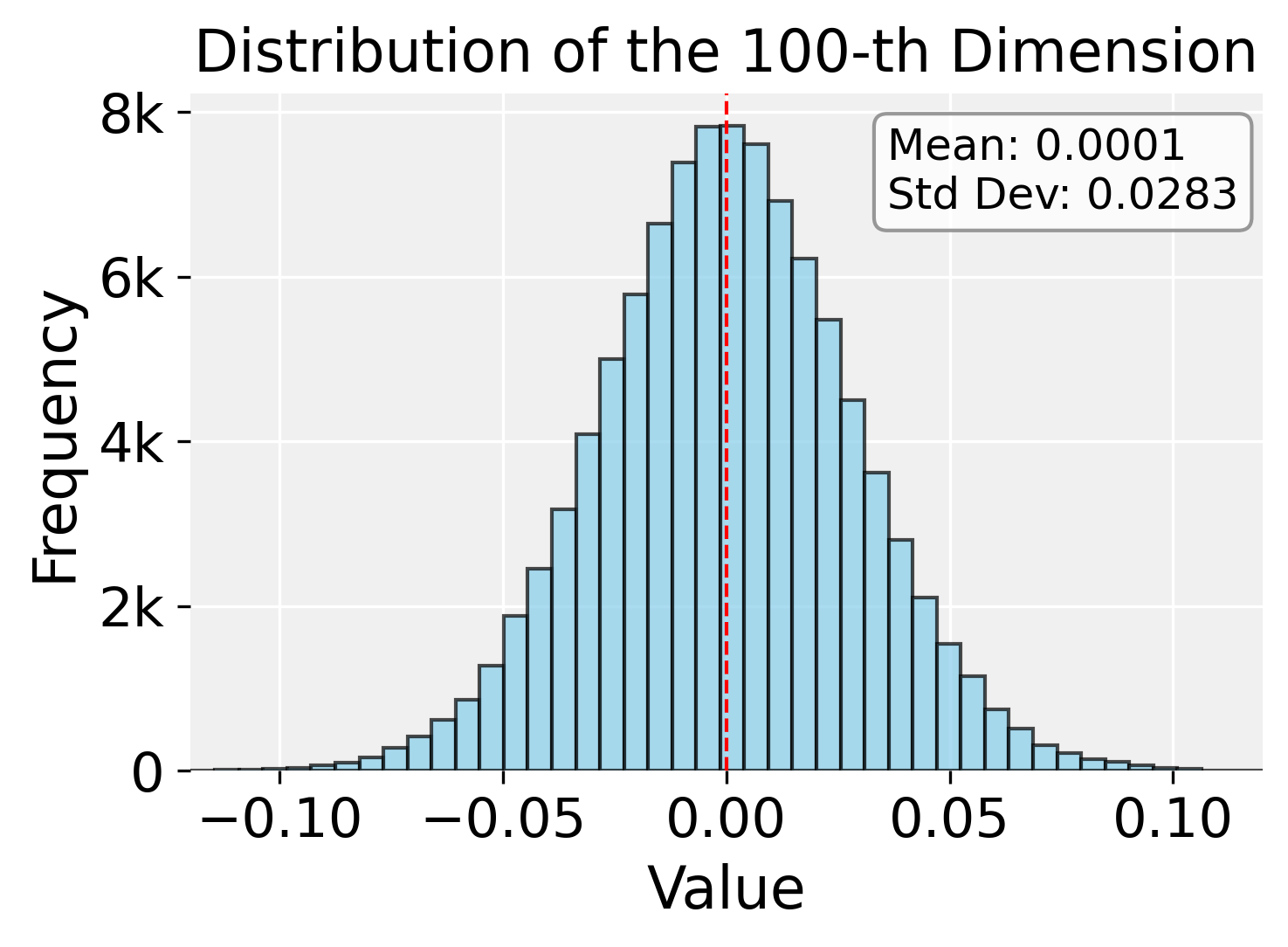}
    \end{minipage}
  }
  \subfigure[OpenAI-1536]{
    \begin{minipage}[t]{0.46\linewidth}
      \centering
      \includegraphics[width=\linewidth]{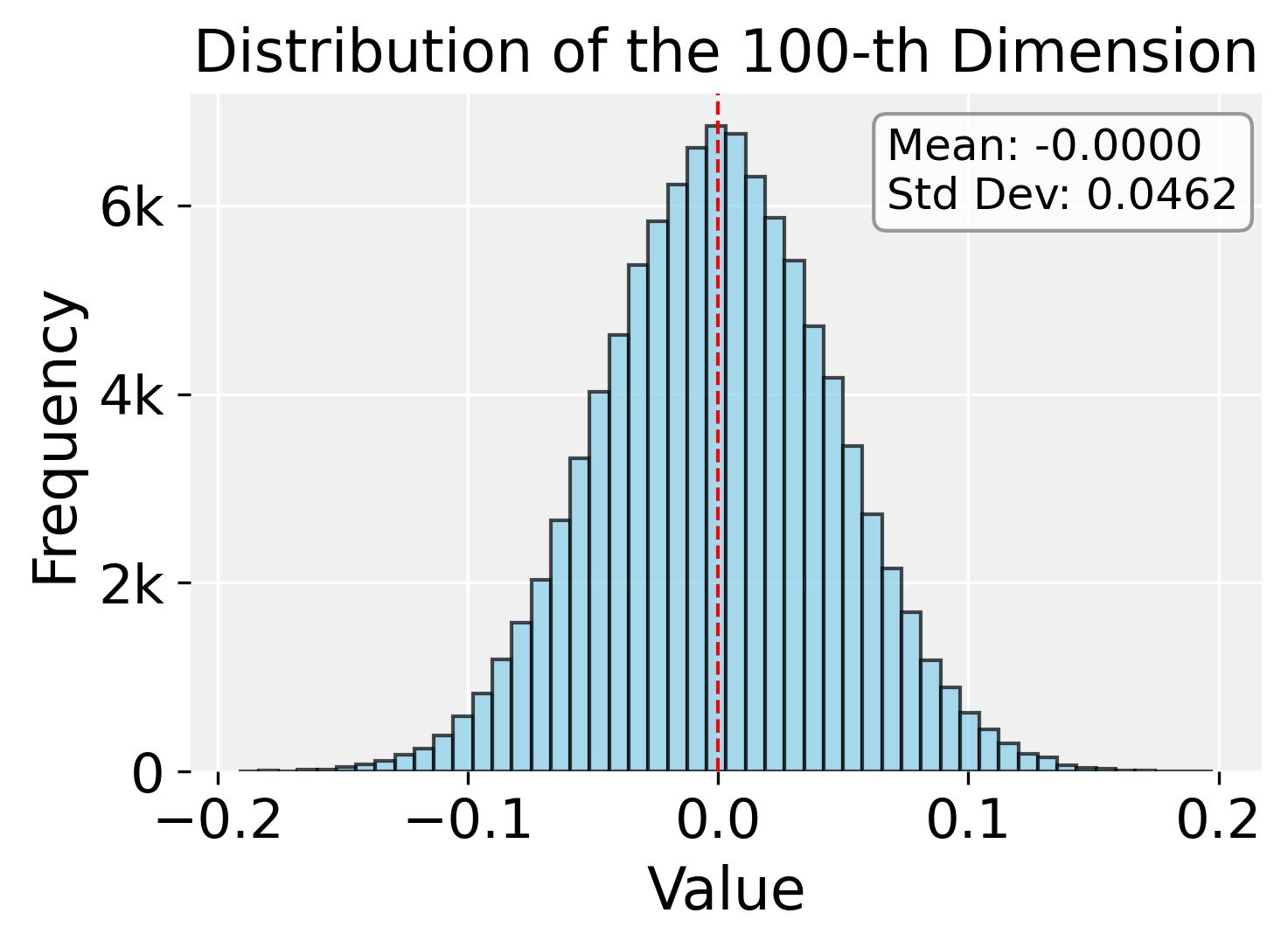}
    \end{minipage}
  }
  \caption{The value distribution at two sampled dimensions  after PCA projection, all vectors are considered.}
  \Description{figure description}
  \label{fig:value-distr}
\end{figure}

\subsection{Multi-Stage Distance Estimation}
\label{sec:saq-estimator}

To process an ANN query, we also split the query vector into segments, estimate the inner product of each segment correspondingly, and combine them to obtain the final distance. Instead of trivially summing up the inner product of each segment, we can also use the feature of the quantization plan to improve the estimator.

Recent studies \cite{gao2023highdimensionalapproximatenearestneighbor} show that reliably identifying the nearest neighbor (NN) does not require exact distance calculations for all candidate vectors. Instead, approximate distance estimates or distance bounds can filter out most unlikely candidates (e.g. if a lower bound exceeds the current best NN distance). In Section \ref{sec:DistEst}, we introduce a progressive distance approximation method that supports progressive refinement for CAQ. For SAQ, we further improve the estimator by leveraging the quantization plan. In particular, the quantization plan partitions the dimensions into segments, and segments with high variance contribute more to the final distance. This inspires us to design a multi-stage estimator that gradually refines the distance estimation by starting with the high-variance segments and gradually adding more segments. Once the distance exceeds the current best NN distance, we can prune the vector.

Moreover, we can even obtain a lower bound of the estimated distance with almost no computation. Specifically, we can predict the distance contribution of each segment with the value distributions of the dimensions in the segment.


For a data vector $\mathbf{o}_r$, the contribution of the dimensions inside segment $\mathit{Seg}$ to the final inner product can be written as
\begin{equation}
  \langle \mathbf{o}_{seg}, \mathbf{q}_{seg} \rangle = \sum_{i \in \mathit{Seg}} \mathbf{q}[i] \cdot \mathbf{o}[i].
\end{equation}
According to our observation, after PCA projection, the value distribution of a dimension almost follows a normal distribution with a mean close to 0, as illustrated in Figure \ref{fig:value-distr}.
The variance of expression above can be written as:
\begin{equation}
  \sigma_{\mathit{Seg}}^2 = Var\left(\langle \mathbf{o}_{seg}, \mathbf{q}_{seg} \rangle \right)
  = \sum_{i \in \mathit{Seg}} \mathbf{q}^2[i] \cdot \sigma_i^2,
\end{equation}
where $\sigma_i^2$ is the variance of dimension $i$ among the dataset.

According to Chebyshev's inequality, we have:
\begin{equation}
  \mathbb{P}\left( \left|\langle \mathbf{o}_{seg}, \mathbf{q}_{seg} \rangle \right| \geq \operatorname{Est}_v(\mathit{Seg}) \right) \leq \frac{1}{m^2},
\end{equation}
where $\operatorname{Est}_v(\mathit{Seg}) = \sigma_{\mathit{Seg}} \cdot m$ and $m$ is a predefined constant. We can use $\operatorname{Est}_v(\mathit{Seg})$ to predict a lower bound of the contribution of each stage to the final inner product with confidence $1-\frac{1}{m^2}$.

It is worth noting that $\sigma_{\mathit{Seg}}$ only needs to be computed once for each query and is shared by all candidate vectors in the query phase. We combine the new estimator $\operatorname{Est}_v(\mathit{Seg})$ with the one in CAQ to conduct a multi-stage estimator, which is a more fine-grained progressive distance refinement provide by SAQ. The multi-stage estimator significantly reduces unnecessary computation and memory access. Our experiments show that the average bit access of the multi-stage estimator is even smaller than the dimension of the dataset, which means that the multi-stage estimator can prune most of the candidates without computing all the segments.

\section{EXPERIMENTAL Evaluation} \label{sec:exp}

\stitle{Baselines}
We compared CAQ and SAQ 
with four other quantization methods listed below.
\squishlist
    \item \textbf{Extended RaBitQ~\cite{extrbq}:} The state-of-the-art quantization method, which quantizes the directions of vectors and stores the lengths. We call it RaBitQ in the experiments.
    \item \textbf{LVQ~\cite{LVQ}:} A recent vector quantization method that achieves good accuracy. LVQ  first collects the minimum value $v_l$ and maximum value $v_r$ of all coordinates for a vector. Then it uniformly divides the range $[v_l, v_r]$ into $2^B - 1$  segments. The floating point value of each coordinate of that vector is rounded to the nearest boundary of the intervals and stored as a $B$-bit integer. 
    \item \textbf{PCA:} We implemented a simple PCA method that first projects all data vectors with a PCA matrix and then quantizes the projected vectors by dropping the insignificant dimensions directly. The dropping rate is equal to the compressing rate.
    \item \textbf{PQ~\cite{5432202}:} A popular quantization method that is generally used with a high compression rate and is widely used in industry. \revised{We use the Faiss PQ implementation~\cite{faiss}. PQ supports two settings $nbits=4$ and $nbits=8$, which is  the number of bits used to represent each sub-vector after quantization. According to ~\cite{extrbq}, $nbits=8$ produces consistently better than $nbits=4$, so we report the result under this setting.}
\squishend

\begin{table}[!t]
  \caption{Datasets.}
  \label{tab:datasets}
  \begin{tabular}{lrrrc}
    \toprule
    \textbf{Dataset} & \textbf{Size} & \textbf{D} & \textbf{Query Size} & \textbf{Type} \\
    \midrule
    DEEP       & 1,000,000  & 256 & 1,000 & Image \\
    GIST        & 1,000,000 & 960 & 1,000 & Image \\
    MSMARC      & 10,000,000 & 1024 & 1,000 & Text \\
    OpenAI-1536 & 999,000   & 1536 & 1,000 & Text \\
    \bottomrule
  \end{tabular}
\end{table}

We do not compare with OPQ~\cite{ge2013optimized}, LSQ~\cite{martinez2016revisiting}, and scalar quantization (SQ) because they are outperformed by RaBitQ. To test the performance of the vector quantization methods for ANNS, we build IVF index~\cite{jegou2010product} for all datasets following RaBitQ. In particular, IVF groups the vectors into clusters and scans the top-ranking clusters for each query. Following the common setting of IVF, we used 4,096 clusters for the datasets.
\revised{
We are aware that proximity graph indexes~\cite{HNSW, NSG} are also popular, but combining with vector indexes is not our focus and we leave it to future work.
All methods were optimized with the SIMD instructions with AVX512 to ensure a fair efficiency comparison, and the performance improvements of SAQ generalize across platforms.}

\begin{figure*}[t]
  \centering
  \includegraphics[width=0.95\linewidth]{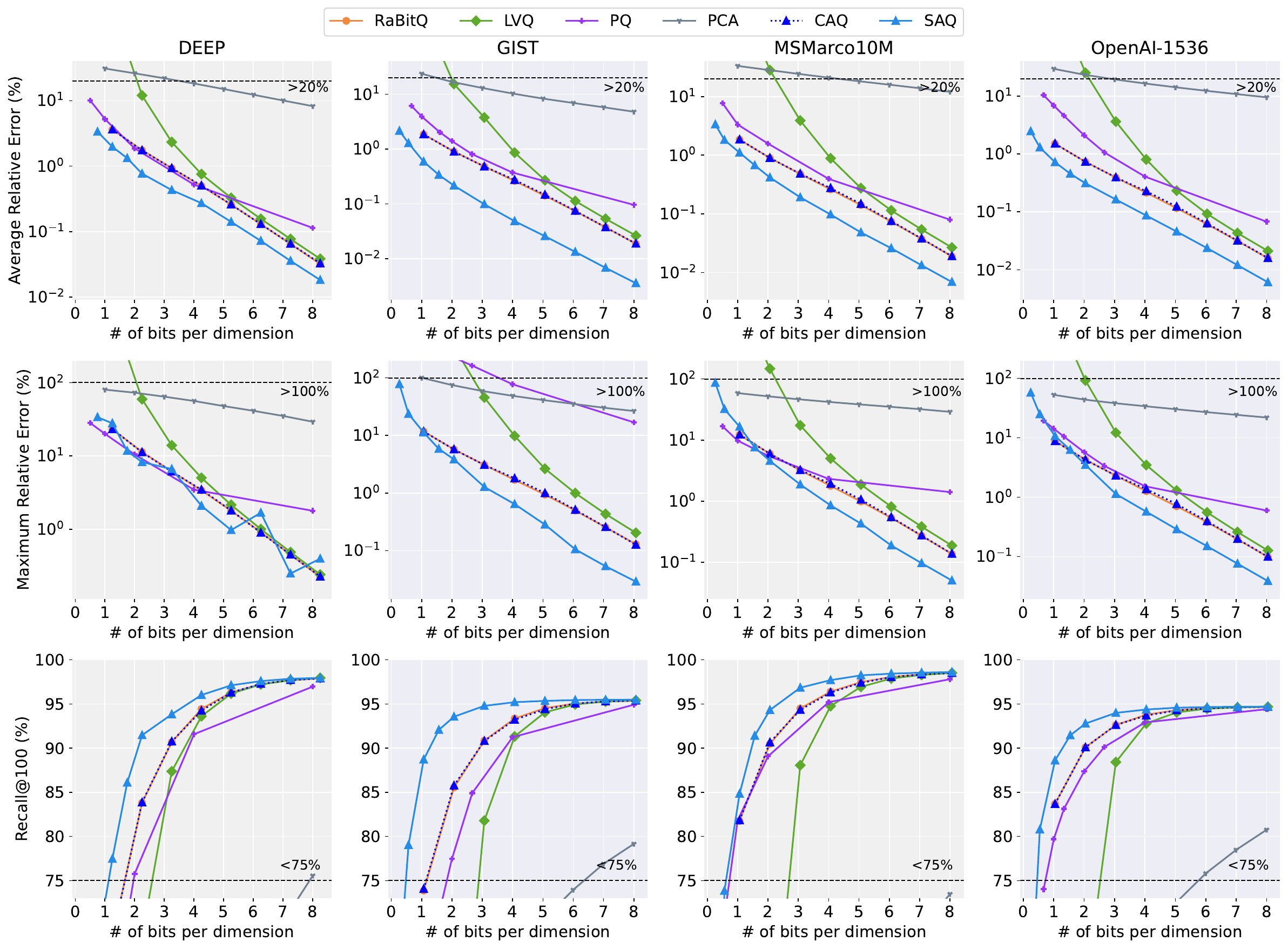}
  \caption{Quantization accuracy of SAQ and the baselines \revised{($nprob=200$)}. }
  \label{fig:error}
\end{figure*}

\stitle{Datasets}
We used four public real-world datasets that have various dimensionalities and data types, as shown in Table \ref{tab:datasets}. These datasets include widely adopted benchmarks for the evaluation of ANN
algorithms\cite{AUMULLER2020101374, 8681160, rbq, extrbq, wangweipca} (DEEP and GIST\footnote{https://www.cse.cuhk.edu.hk/systems/hash/gqr/datasets.html}) and embeddings generated by language models.
The MSMARCO\footnote{https://huggingface.co/datasets/Cohere/msmarco-v2.1-embed-english-v3} dataset contains the embeddings for the TREC-RAG Corpus 2024~\cite{REC-RAG} embedded with the Cohere Embed V3 English model~\cite{CohereV3} and the OpenAI-1536\footnote{https://huggingface.co/datasets/Qdrant/dbpedia-entities-openai3-text-embedding-3-large-1536-1M} dataset is produced
by model \textit{text-embedding-3-large}
of OpenAI~\cite{OpenAI}. We used the query sets provided by the DEEP and GIST datasets. For MSMARCO and OpenAI-1536, we randomly removed 1,000
vectors from each dataset and used them as query vectors.

\begin{table}[]
  \caption{
    \revised{Average relative error at $B\!=\!4$.  SAQ reports the error, while other methods report the error blowup w.r.t. SAQ.} }
  \label{tab:error_comparison}
  \centering
  \small 
  \begin{tabular}{lcccc}
    \toprule
    Method & DEEP & GIST & MSMARC & OpenAI-1536 \\
    \midrule
    SAQ & $0.27\%$ & $0.05\%$ & $0.10\%$ & $0.09\%$ \\
    \midrule
    CAQ (ratio of SAQ)    & $1.9\times$ & $5.6\times$ & $2.8\times$ & $2.6\times$  \\
    RaBitQ (ratio of SAQ) & $1.8\times$ & $5.4\times$ & $2.7\times$ & $2.5\times$  \\
    LVQ (ratio of SAQ)    & $2.8\times$ & $17.8\times$ & $9.1\times$ & $9.3\times$ \\
    PQ  (ratio of SAQ)    & $1.9\times$ & $7.7\times$ & $4.1\times$ & $4.7\times$  \\
    \bottomrule
  \end{tabular}
  \begin{flushleft}
  \end{flushleft}
\end{table}

\stitle{Performance metrics and machine settings}
For the experiments that evaluate the quantization accuracy under different space quota, we measured the accuracy of the estimator in terms of average relative error, maximum relative error, and recall. The relative error is defined as 
$|d^2_{\text{est}} - d^2_{\text{real}}|/d^2_{\text{real}}$, where $d^2_{\text{est}}$ and $d^2_{\text{real}}$ are the estimated and real squared Euclidean distance, receptively. Recall is the percentage of true nearest neighbors successfully retrieved with top $k=100$ and $nprob=200$.
These metrics are widely adopted to measure the accuracy of ANN algorithms~\cite{AUMULLER2020101374, 8681160, 4492921, 10.1145/2213836.2213898, 10.14778/2850469.2850470, 10.1145/3654990}. All metrics were measured on every query and averaged across the entire query set. The compression rate was measured by the number of bits used to quantize a single dimension, which is computed by dividing the total number of bits used to quantize all the dimensions by the number of dimensions. 

\begin{table*}[t]
  \caption{
    \revised{Quantization time  (in seconds) of different methods. \textit{Speedup} is the speedup of SAQ over RaBitQ.} }
  \label{tab:quant_time}
  \centering
  \small
  \begin{tabular}{@{}l *{16}{c}@{}}
    \toprule
    \multirow{2}{*}{Method}
    & \multicolumn{4}{c}{DEEP}
    & \multicolumn{4}{c}{GIST}
    & \multicolumn{4}{c}{MSMARCO}
    & \multicolumn{4}{c}{OpenAI-1536} \\
    \cmidrule(lr){2-5} \cmidrule(lr){6-9} \cmidrule(lr){10-13} \cmidrule(lr){14-17}
    & $B=1$ & $B=4$ & $B=8$ & $B=9$
    & $B=1$ & $B=4$ & $B=8$ & $B=9$
    & $B=1$ & $B=4$ & $B=8$ & $B=9$
    & $B=1$ & $B=4$ & $B=8$ & $B=9$ \\
    \midrule
    RaBitQ        & 0.3 & 3.9 & 21.5 & 41.7 & 2.4 & 16.9 & 83.4 & 165.9 & 33.8 & 178.9 & 897.4 & 1773.6 & 6.1 & 28.2 & 139.7 & 269.0
 \\
    LVQ           & 0.3 & 0.3 & 0.3 & 0.4 & 0.8 & 1.1 & 1.3 & 1.4 & 12.6 & 14.8 & 18.5 & 17.9 & 1.5 & 1.8 & 2.3 & 2.3
\\
    PQ*           & 6.9 & 18.0 & 30.8 & - & 25.1 & 66.3 & 114.1 & - & 144.6 & 256.6 & 396.6 & - & 42.6 & 107.8 & 184.8 & -
\\
    \midrule
    CAQ  & 0.6 & 1.1 & 0.7 & 0.7 & 2.5 & 5.2 & 3.5 & 2.9 & 28.4 & 60.2 & 32.8 & 31.3 & 6.3 & 10.5 & 6.8 & 7.0
 \\
    SAQ  & 0.3 & 0.7 & 0.7 & 0.7 & 2.2 & 3.5 & 2.1 & 2.0 & 27.7 & 38.3 & 26.2 & 26.1 & 3.5 & 7.5 & 3.9 & 3.7
    \\
    \midrule
     \textbf{Speedup} & \textbf{0.9$\times$} & \textbf{5.9$\times$} & \textbf{32.2$\times$} & \textbf{59.1$\times$} & \textbf{1.1$\times$} & \textbf{4.8$\times$} & \textbf{39.1$\times$} & \textbf{84.7$\times$} & \textbf{1.2$\times$} & \textbf{4.7$\times$} & \textbf{34.2$\times$} & \textbf{67.9$\times$} & \textbf{1.7$\times$} & \textbf{3.8$\times$} & \textbf{35.4$\times$} & \textbf{73.1$\times$}   \\
    \bottomrule
  \end{tabular}
  \begin{flushleft}
    \footnotesize \textsuperscript{*}: The results are obtained at the same compression rate. `-' means do not support.
  \end{flushleft}
\end{table*}

For the experiments that evaluate indexing, we measured the index time of different methods using 24 threads. For the experiments that evaluate the performance of ANNS, we measured QPS, i.e., the number of queries processed per second, under different average distance ratios and recalls to show the performance of each method with different compression rates. The average distance ratio is the average of the distance ratios of the retrieved nearest neighbors over the true nearest neighbors. These metrics are widely adopted to measure the performance of ANN algorithms~\cite{AUMULLER2020101374, 10.1145/2213836.2213898, Yusuke-Matsui2018, extrbq}. 

All the experiments were run on a server with 4 Intel Xeon Gold 6252N@ 2.30GHz CPUs (with 24 cores/48 threads each), 756GB RAM. The C++ source code is compiled by GCC 11.4.0 with -Ofast -march=native under Ubuntu 22.04 LTS container with AVX512 enabled. All the results, including indexing and search, were computed using 24 threads bound to 24 physical cores for all methods.

\subsection{Main Results}

\stitle{Quantization accuracy} \revised{In this experiment, we evaluate the accuracy of different quantization methods at different compression rates with fixed $nprob=200$.}  
We focus on both high compression rate ($B=0.2$, compressing $160\times$ approximately) and moderate compression rate ($B=8$, compressing $4\times$ approximately). Note that PQ's compression rate is not controlled by $B$ as in the other methods, but we obtained a very close compression rate for PQ as the other methods for each $B$ in all our experiments.

As reported in Figure~\ref{fig:error}, SAQ (the blue curve) achieves consistently better average relative error and recall than all baseline methods at the same compression rate. The maximum relative error of SAQ is relatively less stable in low-dimension datasets (e.g., DEEP) due to the limitation of SIMD where the granularity of segmentation is  at least 64, which leads to a suboptimal quantization plan. Nevertheless, as the dimension increases, this limitation is hidden and both the average and  maximum relative errors of  SAQ  decrease exponentially with the number of bits. Additionally, even our base method CAQ also consistently achieves almost the same performance as RaBitQ in this experiment, which shows that our new efficient quantization method does not compromise accuracy.

At higher compression rates ($B\leq 4$), the accuracy of the LVQ worsens dramatically as the compression rate increases. Table~\ref{tab:error_comparison} reports the error comparison of different methods at $B=4$. It is worth noting that as shown in Figure~\ref{fig:error},  SAQ can arealdy achieve a high recall of  $>95\%$ at $B=4$, and in some case like for MSMARCO, even $B=2$ is good enough.

While the maximum compression rate of RaBitQ and LVQ is around $32\times$ (at $B=1$),  SAQ and also PQ can achieve higher compression rates with reasonable accuracy. However, for high compression rates, e.g., when $B=0.5$ ($\sim 64\times$ compression), the average relative error of SAQ is significantly better than that of PQ (e.g., $4.8\times$ lower for GIST), and impressively, is even consistently lower than the error of the state-of-the-art RaBitQ obtained at a much lower compression rate when $B=1$ ($\sim 32\times$ compression).  In a more extreme case $B=0.2$ ($\sim 160\times$ compression), SAQ still achieves better accuracy than PQ. At lower compression rates ($B> 4$), SAQ also consistently achieves the lowest average and maximum relative errors and recalls compared to baselines. The average relative error is $2.5\times$ to $5\times$ lower than that of RaBitQ in high-dimensional datasets.

\stitle{Quantization efficiency}
In this experiment, we evaluate the quantization efficiency of different methods. The quantization time includes generating the random orthogonal matrix and applying the rotation to all the data vectors, but excludes the time for applying the PCA projection. Since the PCA projection and the random rotation can be combined into one matrix, this will not change the time complexity.
As reported in Table~\ref{tab:quant_time}, the quantization time of RaBitQ increases exponentially with $B$. When $B=9$, RaBitQ already costs 11.8 CPU hours (i.e., 0.49 hours with 24 threads) to quantize MSMARCO, which contain 10M vectors with 1,024 dimensions. 
In contrast, the quantization time of CAQ and SAQ fluctuates only within a manageable time range, similar to that of LVQ. The speedup of SAQ over RaBitQ can be over $80\times$ at $B=9$.

\begin{figure*}[]
  \centering
  \includegraphics[width=0.95\linewidth]{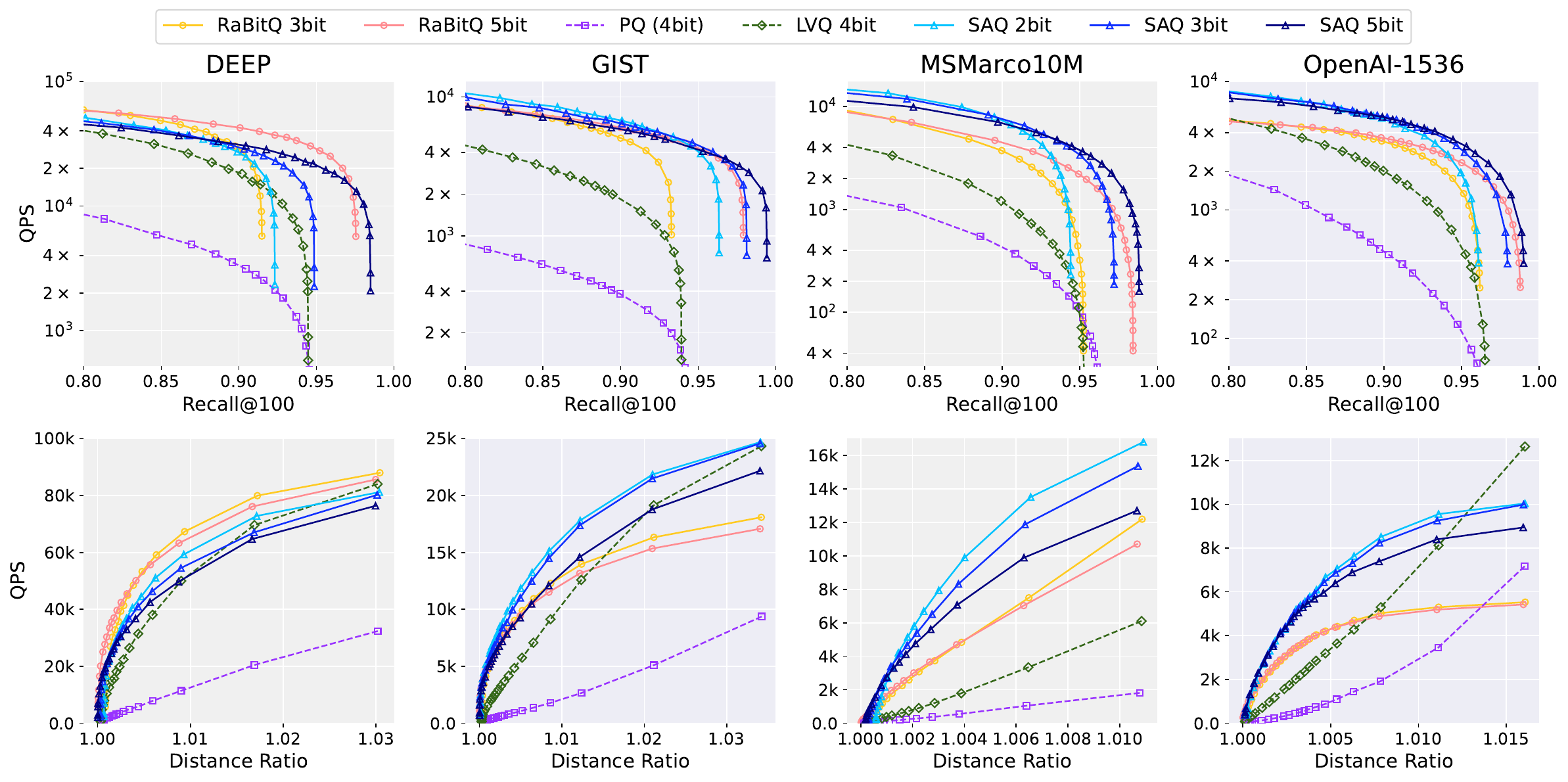}
  \caption{\revised{The performance of SAQ and the baselines for ANNS (higher QPS/recall and lower average distance ratio are better).}}
  \Description{figure description}
  \label{fig:qps}
\end{figure*}

\stitle{ANNS performance}
In this experiment, we evaluate the performance of CAQ, SAQ, and baselines to process ANNS queries. Figure \ref{fig:qps} reports the QPS-Recall curves (upper right, i.e., higher QPS/recall, is better) and the QPS-Ratio curves (upper left, i.e., higher QPS and lower average distance ratio, is better) for both methods. Using the IVF index, as the number of clusters to probe increases, QPS decreases as computation increases, but the probability of finding the ground truth increases. A lower compression rate can produce a more accurate estimated distance to reach a higher recall, but requires more computation that slows down the distance estimator. We set $m=2$ for the multi-stage estimator of SAQ in GIST and $m=4$ in other datasets. \revised{We report the results of SAQ for $B=\{2,3,5\}$, RaBitQ for $B=\{3,5\}$ (it does not support $B=2$), LVQ for $B=4$ and PQ for the compression rate $8\times$ (equivalent to $4$bit).}

\revised{
When the number of clusters to probe is low, LVQ can achieve high QPS, especially for the high-dimensional datasets, since it does not require the rotation of the query vectors. However, as the number of clusters to probe increases, the QPS of LVQ drops significantly due to the high computation cost of a single distance estimation and can not reach the recall as $3$bit RaBitQ and SAQ, even $2$bit SAQ.
}

\begin{table}[t]
  \caption{QPS at $95\%$ recall for SAQ and RaBitQ. `-' means cannot reach $95\%$ recall, best QPS in bold for each method.}
  \label{tab:qps}
  \centering
  \small
  \begin{tabular}{@{}l *{8}{c}@{}}
    \toprule
    \multirow{2}{*}{Method}
    & \multicolumn{2}{c}{DEEP}
    & \multicolumn{3}{c}{GIST}\\
    \cmidrule(lr){2-3} \cmidrule(lr){4-6}
     & $B=3$ & $B=5$
     & $B=2$ & $B=3$ & $B=5$ \\
    \midrule
    RaBitQ & -- &  \textbf{27721}     & -- & -- & \textbf{4018}
    \\
    SAQ  &   --\textit{*} & \textbf{19407}    & 3915 & \textbf{4683} & 4073
    \\

    \bottomrule
    \toprule

    \multirow{2}{*}{Method}
    & \multicolumn{2}{c}{MSMARCO}
    & \multicolumn{3}{c}{OpenAI-1536} \\
    \cmidrule(lr){2-3} \cmidrule(lr){4-6}
    & $B=3$ & $B=5$
    & $B=2$ & $B=3$ & $B=5$  \\
    \midrule
    RaBitQ   & 235 & \textbf{1958}   & --  & 1338 & \textbf{2329}
    \\
    SAQ      & 3427 & \textbf{3726}   & 1624 & 2812 & \textbf{3112}
    \\
    \bottomrule
  \end{tabular}

  \begin{flushleft}
  \footnotesize \textit{*} The maximum recall of SAQ in DEEP is 94.86\% when $B=3$.
  \end{flushleft}
\end{table}

\revised{
For SAQ, in the low-dimensional dataset (DEEP), it can achieve a higher recall than RaBitQ at the same compression rate.
In fact, $2$bit SAQ attains a higher recall than $3$bit RaBitQ.} However, the QPS of SAQ is hindered by the additional computation of the multi-stage estimator. This is because for low-dimensional data, the computation of an individual estimator is minor and the effect of the early pruning of candidates is negligible, while the constant overhead of the multi-stage estimator is more significant.

The multi-stage estimator is more suitable for high-dimensional datasets, where the computation of a single estimator is heavier and the overhead of the multi-stage estimator becomes minor. Early pruning significantly reduces computation and memory access.
Thus, for high-dimension datasets, SAQ consistently outperforms RaBitQ in terms of QPS, recall, and average distance ratio at the same compression rate. Even at $2$bit SAQ can still almost outperform $3$bit RabitQ and produces $\sim 95\%$ recall.

In Table \ref{tab:qps}, we also report the QPS for different datasets and configurations when 95\% recall is achieved. For GIST and OpenAI-1536, SAQ can achieve $>95\%$ recall and reasonable QPS with $B=2$. For MSMARCO, SAQ can also produce a maximum $94.4\%$ recall with $B=2$. Thus, $B=3$ is sufficient for SAQ to produce 95\% recall for these high-dimensional datasets, and $B=5$ almost gives the best QPS. We notice that for MSMARCO, the QPS of SAQ is $14.6\times$ that of RaBitQ  when $B=3$. This significant difference occurs because RaBitQ approaches its maximum recall at $B=3$ due to error limitations, causing a sharp decline in its QPS, while SAQ maintains high performance with room for further recall improvement.

\subsection{Micro Results}

\label{sec:micro_results}

\stitle{Space consumption}
We first report the space consumption of SAQ and RaBitQ with different configurations $B$. The space includes the quantization code, two  extra factors per data vector and other factors that consume constant space. We only report the results from the MSMARCO dataset under a few configurations of $B$. The space consumption is proportional to the number of data vectors and the number of dimensions for the other datasets, and also proportional to other configurations of  $B$.

As reported in Table\ref{tab:quant_space}, the space consumption is almost proportional to $B$. We note that, due to the constant overhead of some factors, the compression rate may not be as expected in small $B$. The space consumption of SAQ is slightly larger than the baseline since the multi-stage estimator of SAQ requires to store more statistics of the dataset and consume more space for extra factors. However, this overhead is negligible for large-scale datasets.

\begin{table}[]
  \caption{Storage space of the quantized vectors for MSMARCO. `-': the configuration not supported by the method.}
  \label{tab:quant_space}
  \centering
  \small
  \begin{tabular}{l*{10}{r}}
  \toprule
  \textbf{B}  & 0.5  & 1  & 2  & 4  & 6   & 8  \\
  \midrule
  RaBitQ (MB)    & -- & 1363 & -- & 5025 & -- & 9908
   \\
  CAQ (MB)    & -- & 1379 & 2600 & 5041 & 7483 & 9924
   \\
  SAQ (MB)    & 838 & 1449 & 2671 & 5114 & 7481 & 9922
  \\
  \bottomrule
  \end{tabular}
  \begin{flushleft}
    \footnotesize Raw data vector space consumption is 39,100 MB.
  \end{flushleft}
\end{table}

\begin{figure}[!]
  \centering
  \includegraphics[width=\linewidth]{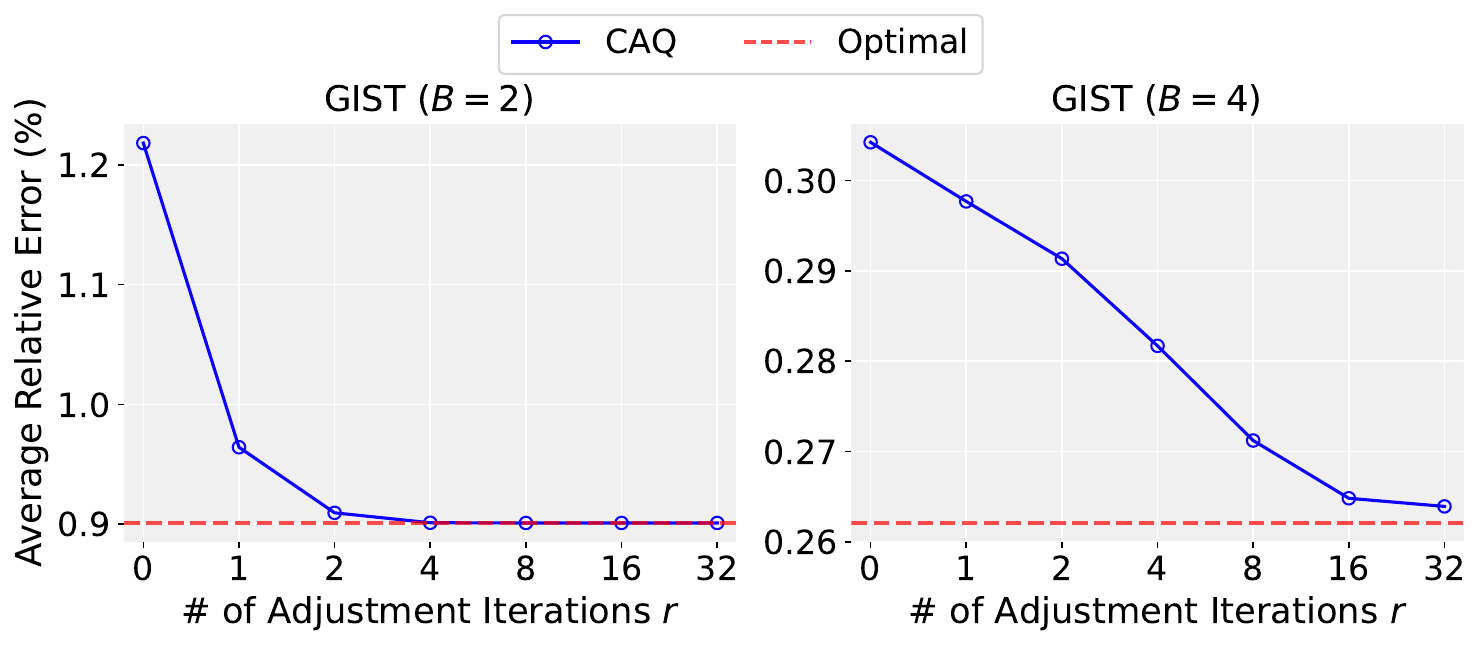}
  \caption{\revised{Quantization accuracy with different number of code adjustment iterations $r$. \textit{Optimal} means the error produced by the RabitQ quantization code, which is optimal.}}
  \label{fig:code-adj-acc}
\end{figure}

\stitle{Code adjustment iteration}
\revised{
We compare the quantization accuracy with different number of code adjustment iterations $r$ in Algorithm~\ref{alg:find-adjusted-quantization}. We measure accuracy by the average relative error, which is the same as in our main results.}

\revised{
As reported in Figure~\ref{fig:code-adj-acc}, without any code adjustment ($r=0$), the average relative error is the worst. As $r$ increases, the average relative errors become significantly better and almost optimal (the error is only $0.7\%$ worse than the optimal when $B=4, r=32$). The results show that the first few rounds of iterations can significantly refine the quantization code and produce a good enough error while increasing $r$ to a higher value only slightly improves accuracy. Thus, we recommend setting $r\in[4, 8]$ for CAQ, which can achieve a good balance between accuracy and efficiency.}

\stitle{Memory access for multi-stage estimator} We also quantify the memory access and computational costs of the multi-stage estimator of SAQ. \revised{We measure the average bits accessed by the multi-stage estimator when processing an ANNS query, which is the number of bits of quantization codes required by the estimator for a single distance estimation.} As defined in Section \ref{sec:saq-estimator}, the estimator confidence level is governed by the parameter $m$ in $\operatorname{Est}_v(\mathit{Seg}) = \sigma_{\mathit{Seg}} \cdot m$. As $m$ increases, the confidence in the lower bound distance produced by the estimator increases, though it is likely that  candidates are pruned. When setting $m$ to a small value,  pruning becomes radical and it may weaken the recall.

As reported in Figure \ref{fig:multi-stage-est}, the average number of bits accessed by a query is always larger than or equal to the number of dimensions of the dataset. This is because RaBitQ always estimates the distance with the most significant bit of the quantization code. RaBitQ also uses this distance to prune candidates, and so the bits accessed increases slowly as $B$ (i.e., the number of bits for a dimension) increases.
For SAQ's multi-stage estimator, the memory access is consistently reduced as $m$ decreases, which means that candidates are pruned earlier before accessing more bits and computing a more accurate distance.
We also observe that the multi-stage estimator almost does not harm the recall when $m\ge4$, but the average of bits accessed increases correspondingly for larger $m$.  When $m=2$, although the bits accessed decreases dramatically, the recall is also lower by the more radical pruning. Thus, we recommend setting $m=4$ as default for the multi-stage estimator in SAQ, which can significantly reduce bits accessed, while avoiding weakening recall. SAQ's multi-stage estimator using default $m$ can reduce the bits accessed by about 1.9-4.0$\times$ compared with the estimator of RaBitQ.

\begin{figure}[!]
  \centering
  \includegraphics[width=\linewidth]{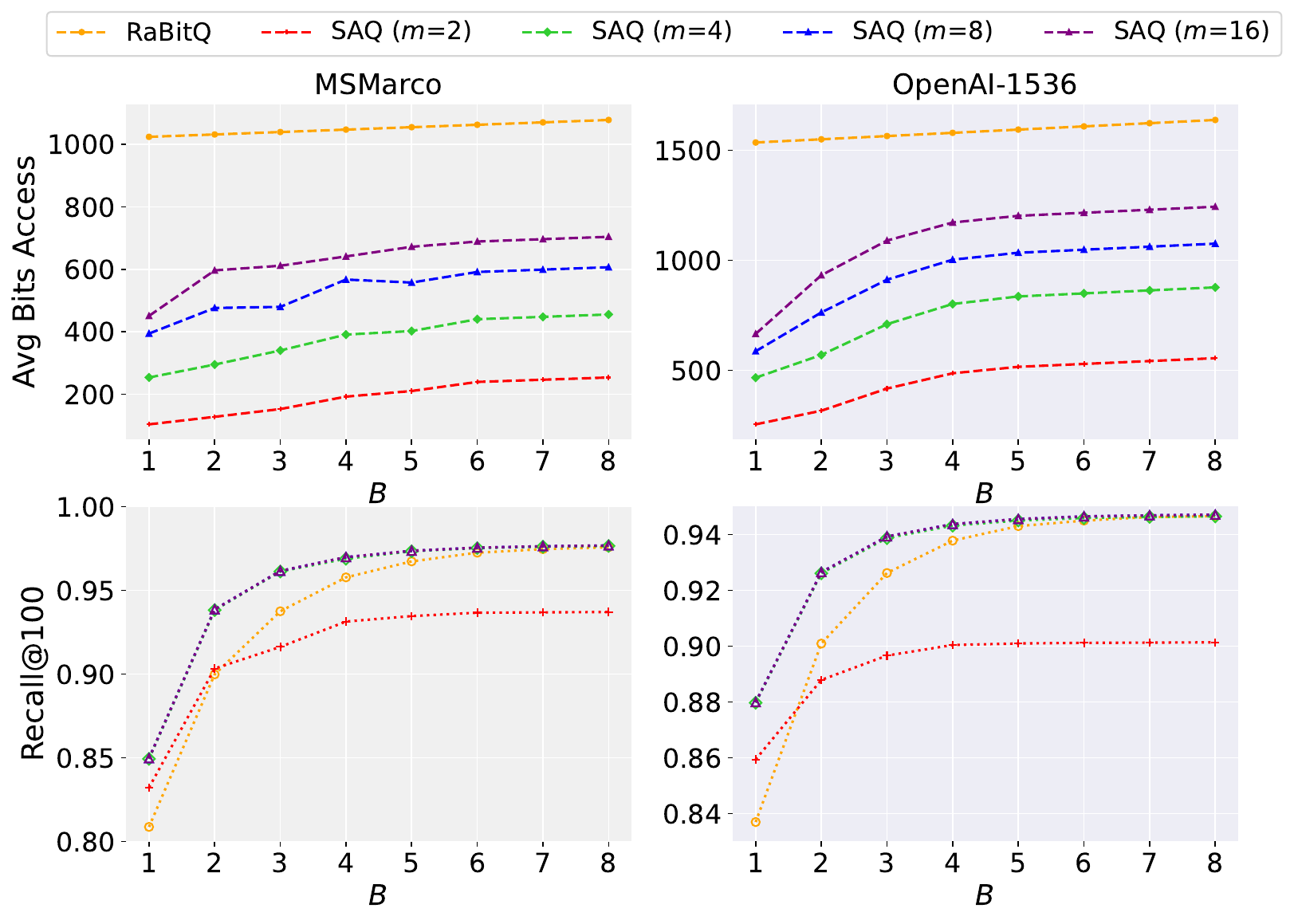}
  \caption{\revised{The average number of bits accessed for each vector (top, smaller the better) and query recall (bottom, $nprob=200$) for the estimator of RaBitQ and multi-stage estimator of SAQ with different $m$ under different $B$ (the recall curves of $m=\{4, 8, 16\}$ are almost identical).}}
  \Description{figure description}
  \label{fig:multi-stage-est}
\end{figure}

\begin{figure}[!]
  \centering
  \includegraphics[width=\linewidth]{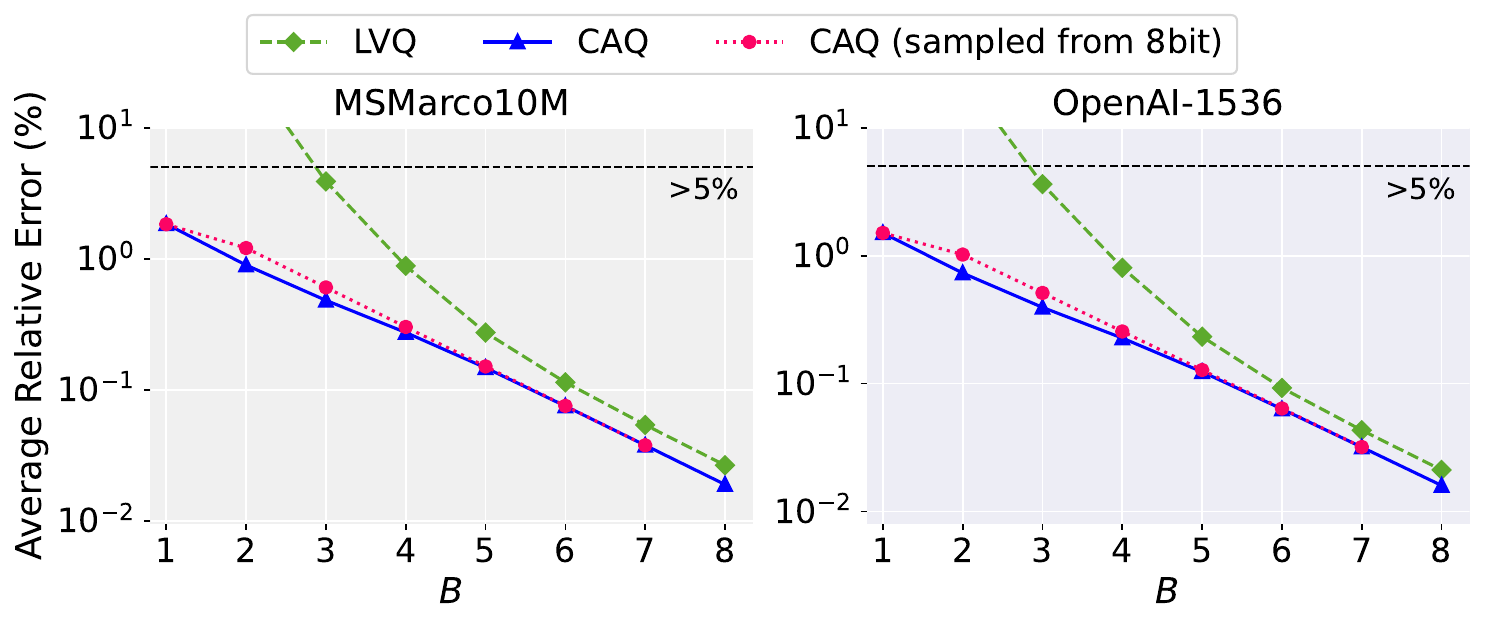}
  \caption{Progressive
  distance approximation: $b$ bits sampled from $8$-bit CAQ and compared with native $b$-bit CAQ.}
  \Description{figure description}
  \label{fig:progressive-error}
\end{figure}

\stitle{Accuracy of progressive distance approximation}
\label{sec:progressive-error}
We compare the relative error produced by the progressive distance (red curve) with the error of the native distance of CAQ (blue curve) and LVQ (green curve). The quantization code for progressive distance is sampled from the native 8-bit quantization code to different $b$ bits.


As reported in Figure \ref{fig:progressive-error}, the error of the sampled quantization code is consistently lower than that of LVQ and is almost the same as that of the native quantization code when $b > 4$ and the error is slightly larger  when $2 \le b \le 4$.  The small increase in error when reusing the factor under small $b \le 4$ because the $b$-bit code is sampled from the 8-bit quantization code and the estimator factor is optimized to fit the full 8-bit code, which enlarges the gap between the reuse factor and the native factor when more bits of code are dropped. However, when $b=1$, the error of the sampled quantization code and the native one is almost the same, since the estimator factor is highly concentrated around 0.8 when $b=1$~\cite{rbq} and the estimator does not rely on the factor produced by 8-bit quantization.


\section{Related Work}

\stitle{Clustering-based vector quantization} Product quantization (PQ) is a classic method that divides the $D$-dimensional vector space into $M$ sub-spaces with dimension $D/M$ and uses K-means clustering to obtain $K$ vector codewords for each sub-space~\cite{babenko2014additive, jegou2010product, martinez2018lsq++, wang2017survey}. The codewords for each sub-space forms a codebook, and a PQ has $M$ codebooks. Optimized PQ (OPQ) improves PQ by applying a rotation matrix $R$ beforehand to evenly distribute the energy among the $M$ sub-spaces~\cite{ge2013optimized}. Locally optimized PQ (LOPQ) works with the IVF index and trains a set of codebooks for each cluster of vectors to better adapt to local data distribution~\cite{kalantidis2014locally}. Instead of partitioning a vector dimension as in PQ, residue quantization (RQ) uses codebooks that are in the same $D$-dimensional  space of the original vectors~\cite{babenko2014additive, jegou2010product, martinez2018lsq++, wang2017survey}. The codebooks are trained sequentially with K-means with each codebook working on the residues from all its previous codebooks. To improve accuracy, additive quantization (AQ) improves RQ by jointly learning the $M$ $D$-dimensional codebooks~\cite{babenko2014additive}. However, the complexity is high for both learning the codebooks and encoding each vector with the codebooks. To speed up AQ, LSQ~\cite{martinez2016revisiting} and LSQ++~\cite{martinez2018lsq++} conduct both algorithm and system optimizations by reformulating the codebook learning problem and using GPU for acceleration. There are also vector quantization methods with other codebook structures, e.g., composite quantization (CQ)~\cite{wang2018composite} and tree quantization (TreeQ)~\cite{foote2003treeq}.

Clustering based methods use lookup tables (LUTs) for efficient approximate distance computation~\cite{babenko2014additive, jegou2010product, martinez2018lsq++, wang2017survey}. A query first computes its distances to all the codewords to initialize the LUTs, and to compute the approximate distance with a vector $x$, the distances of $x$'s codewords are aggregated over the codebooks. Typically, a codebook size of $K=256$ is used such that the index of each codeword can be stored as a byte. However, it is observed that table lookup does not utilize CPU registers efficiently with $K=256$~\cite{andre2016cache}, and thus it is proposed to use smaller codebooks (i.e., $K=16$) to fit each LUT into the registers~\cite{sun2023soar}. However, at the same space quota for each vector, smaller codebooks degrade accuracy.

\stitle{Projection-based vector quantization} It is well known that the energy of high-dimension vectors concentrate on the leading dimensions after projection with the PCA matrix. As such, FEXIPRO~\cite{li2017fexipro} adopts the leading dimensions to compute similarity bounds for efficient pruning in maximum inner product search (MIPS). DADE uses PCA-based dimension reduction for Euclidean distance and derives probabilistic bounds for the approximate distance computed with the leading dimensions~\cite{deng2024efficient}. A progressive distance computation method is introduced, where a vector first uses its leading dimensions for computation and then checks whether its distance to query is smaller than the current top-$k$ with high probability, and exact distance is only computed if the check passes. In a similar vein, ADSampling projects the vectors with a random orthonormal projection matrix and shows that a unbiased distance estimation can be obtained by sampling some dimensions~\cite{deng2024efficient}. LeanVec learns the projection matrix for dimension reduction in a query-aware manner~\cite{tepper2023leanvec}, while GleanVec~\cite{tepper2024gleanvec} improves LeanVec by learning multiple projection matrices to adapt to local data distributions like LOPQ. Different from the dimension reduction methods, RabitQ uses the random orthonormal projection matrix such that vector dimensions have similar magnitude and quantizes each dimension with 1-bit~\cite{rbq}. To use more bits for each dimension, RabitQ pads the original vector with zeros to increase the dimension before projection. Observing that padding yields  inferior accuracy, extended RabitQ (E-RabitQ) uses codewords on the unit sphere after projection and is shown to be asymptotically optimal in accuracy~\cite{extrbq}. Currently, E-RabitQ represents the state-of-the-art in accuracy. As we demonstrated in Section~\ref{sec:exp}, SAQ significantly improves E-RabitQ in terms of both encoding speed and quantization  accuracy.

\section{Conclusions}


We have presented SAQ, an advanced vector quantization method that addresses critical limitations of existing methods in both encoding efficiency and quantization accuracy through dimension segmentation and code adjustment. Extensive experiments demonstrate that SAQ significantly outperforms state-of-the-art methods, achieving up to 80\% lower quantization error and 80× faster encoding speeds compared to Extended RabitQ. These advancements position SAQ as a robust solution for high-accuracy, efficient vector quantization in large-scale ANNS applications.

\bibliographystyle{ACM-Reference-Format}
\bibliography{main}

@article{rbq,
  author = {Gao, Jianyang and Long, Cheng},
  title = {RaBitQ: Quantizing High-Dimensional Vectors with a Theoretical Error Bound for Approximate Nearest Neighbor Search},
  year = {2024},
  issue_date = {June 2024},
  publisher = {Association for Computing Machinery},
  address = {New York, NY, USA},
  volume = {2},
  number = {3},
  url = {https://doi.org/10.1145/3654970},
  doi = {10.1145/3654970},
  journal = {Proc. ACM Manag. Data},
  month = may,
  articleno = {167},
  numpages = {27}
}

@misc{extrbq,
      title={Practical and Asymptotically Optimal Quantization of High-Dimensional Vectors in Euclidean Space for Approximate Nearest Neighbor Search},
      author={Jianyang Gao and Yutong Gou and Yuexuan Xu and Yongyi Yang and Cheng Long and Raymond Chi-Wing Wong},
      year={2024},
      eprint={2409.09913},
      archivePrefix={arXiv},
      primaryClass={cs.DB},
      url={https://arxiv.org/abs/2409.09913},
}

@article{LVQ,
author = {Aguerrebere, Cecilia and Bhati, Ishwar Singh and Hildebrand, Mark and Tepper, Mariano and Willke, Theodore},
title = {Similarity Search in the Blink of an Eye with Compressed Indices},
year = {2023},
issue_date = {July 2023},
publisher = {VLDB Endowment},
volume = {16},
number = {11},
issn = {2150-8097},
url = {https://doi.org/10.14778/3611479.3611537},
doi = {10.14778/3611479.3611537},
journal = {Proc. VLDB Endow.},
month = jul,
pages = {3433–3446},
numpages = {14}
}

@misc{wangweipca,
      title={Fast High-dimensional Approximate Nearest Neighbor Search with Efficient Index Time and Space},
      author={Mingyu Yang and Wentao Li and Wei Wang},
      year={2025},
      eprint={2411.06158},
      archivePrefix={arXiv},
      primaryClass={cs.DB},
      url={https://arxiv.org/abs/2411.06158},
}

@misc{yin2025gorgeousrevisitingdatalayout,
      title={Gorgeous: Revisiting the Data Layout for Disk-Resident High-Dimensional Vector Search}, 
      author={Peiqi Yin and Xiao Yan and Qihui Zhou and Hui Li and Xiaolu Li and Lin Zhang and Meiling Wang and Xin Yao and James Cheng},
      year={2025},
      eprint={2508.15290},
      archivePrefix={arXiv},
      primaryClass={cs.DB},
      url={https://arxiv.org/abs/2508.15290}, 
}

@misc{gao2023highdimensionalapproximatenearestneighbor,
      title={High-Dimensional Approximate Nearest Neighbor Search: with Reliable and Efficient Distance Comparison Operations},
      author={Jianyang Gao and Cheng Long},
      year={2023},
      eprint={2303.09855},
      archivePrefix={arXiv},
      primaryClass={cs.DS},
      url={https://arxiv.org/abs/2303.09855},
}

@INPROCEEDINGS{8104097,
  author={Alon, Noga and Klartag, Bo'az},
  booktitle={2017 IEEE 58th Annual Symposium on Foundations of Computer Science (FOCS)},
  title={Optimal Compression of Approximate Inner Products and Dimension Reduction},
  year={2017},
  volume={},
  number={},
  pages={639-650},
  doi={10.1109/FOCS.2017.65}}

@online{faiss,
  author="Faiss",
  title="{Faiss}",
  url="https://github.com/facebookresearch/faiss",
  note="(2023)",
}

@online{REC-RAG,
  author="TREC-RAG",
  title="TREC-RAG Corpus 2024",
  url="https://trec-rag.github.io/annoucements/2024-corpus-finalization/",
  note="(2024)",
}

@online{CohereV3,
  author="Cohere",
  title="Cohere Embed V3",
  url="https://cohere.com/blog/introducing-embed-v3",
  note="(2023)",
}

@online{OpenAI,
  author="OpenAI",
  title="New embedding models and API updates.",
  url="https://openai.com/index/new-embedding-models-and-api-updates/",
  note="(2024)",
}

@article{AUMULLER2020101374,
title = {ANN-Benchmarks: A benchmarking tool for approximate nearest neighbor algorithms},
journal = {Information Systems},
volume = {87},
pages = {101374},
year = {2020},
issn = {0306-4379},
doi = {https://doi.org/10.1016/j.is.2019.02.006},
url = {https://www.sciencedirect.com/science/article/pii/S0306437918303685},
author = {Martin Aumüller and Erik Bernhardsson and Alexander Faithfull}
}

@ARTICLE{5432202,
  author={Jégou, Herve and Douze, Matthijs and Schmid, Cordelia},
  journal={IEEE Transactions on Pattern Analysis and Machine Intelligence},
  title={Product Quantization for Nearest Neighbor Search},
  year={2011},
  volume={33},
  number={1},
  pages={117-128},
  doi={10.1109/TPAMI.2010.57}}

@ARTICLE{8681160,
  author={Li, Wen and Zhang, Ying and Sun, Yifang and Wang, Wei and Li, Mingjie and Zhang, Wenjie and Lin, Xuemin},
  journal={IEEE Transactions on Knowledge and Data Engineering},
  title={Approximate Nearest Neighbor Search on High Dimensional Data — Experiments, Analyses, and Improvement},
  year={2020},
  volume={32},
  number={8},
  pages={1475-1488},
  doi={10.1109/TKDE.2019.2909204}}

@INPROCEEDINGS{4492921,
  author={Patella, Marco and Ciaccia, Paolo},
  booktitle={First International Workshop on Similarity Search and Applications (sisap 2008)},
  title={The Many Facets of Approximate Similarity Search},
  year={2008},
  volume={},
  number={},
  pages={10-21},
  doi={10.1109/SISAP.2008.18}}

@inproceedings{10.1145/2213836.2213898,
author = {Gan, Junhao and Feng, Jianlin and Fang, Qiong and Ng, Wilfred},
title = {Locality-sensitive hashing scheme based on dynamic collision counting},
year = {2012},
isbn = {9781450312479},
publisher = {Association for Computing Machinery},
address = {New York, NY, USA},
url = {https://doi.org/10.1145/2213836.2213898},
doi = {10.1145/2213836.2213898},
booktitle = {Proceedings of the 2012 ACM SIGMOD International Conference on Management of Data},
pages = {541–552},
numpages = {12},
location = {Scottsdale, Arizona, USA},
series = {SIGMOD '12}
}

@article{10.14778/2850469.2850470,
author = {Huang, Qiang and Feng, Jianlin and Zhang, Yikai and Fang, Qiong and Ng, Wilfred},
title = {Query-aware locality-sensitive hashing for approximate nearest neighbor search},
year = {2015},
issue_date = {September 2015},
publisher = {VLDB Endowment},
volume = {9},
number = {1},
issn = {2150-8097},
url = {https://doi.org/10.14778/2850469.2850470},
doi = {10.14778/2850469.2850470},
journal = {Proc. VLDB Endow.},
month = sep,
pages = {1–12},
numpages = {12}
}

@article{10.1145/3654990,
author = {Su, Yongye and Sun, Yinqi and Zhang, Minjia and Wang, Jianguo},
title = {Vexless: A Serverless Vector Data Management System Using Cloud Functions},
year = {2024},
issue_date = {June 2024},
publisher = {Association for Computing Machinery},
address = {New York, NY, USA},
volume = {2},
number = {3},
url = {https://doi.org/10.1145/3654990},
doi = {10.1145/3654990},
journal = {Proc. ACM Manag. Data},
month = may,
articleno = {187},
numpages = {26}
}

@article{Yusuke-Matsui2018,
  title={[Invited Paper] A Survey of Product Quantization},
  author={Yusuke Matsui and Yusuke Uchida and Herv\&eacute; J\&eacute;gou and Shin'ichi Satoh},
  journal={ITE Transactions on Media Technology and Applications},
  volume={6},
  number={1},
  pages={2-10},
  year={2018},
  doi={10.3169/mta.6.2}
}

@inproceedings{devlin2019bert,
  title={Bert: Pre-training of deep bidirectional transformers for language understanding},
  author={Devlin, Jacob and Chang, Ming-Wei and Lee, Kenton and Toutanova, Kristina},
  booktitle={Proceedings of the 2019 conference of the North American chapter of the association for computational linguistics: human language technologies, volume 1 (long and short papers)},
  pages={4171--4186},
  year={2019}
}

@inproceedings{radford2021learning,
  title={Learning transferable visual models from natural language supervision},
  author={Radford, Alec and Kim, Jong Wook and Hallacy, Chris and Ramesh, Aditya and Goh, Gabriel and Agarwal, Sandhini and Sastry, Girish and Askell, Amanda and Mishkin, Pamela and Clark, Jack and others},
  booktitle={International conference on machine learning},
  pages={8748--8763},
  year={2021},
  organization={PmLR}
}

@inproceedings{shvetsova2022everything,
  title={Everything at once-multi-modal fusion transformer for video retrieval},
  author={Shvetsova, Nina and Chen, Brian and Rouditchenko, Andrew and Thomas, Samuel and Kingsbury, Brian and Feris, Rogerio S and Harwath, David and Glass, James and Kuehne, Hilde},
  booktitle={Proceedings of the ieee/cvf conference on computer vision and pattern recognition},
  pages={20020--20029},
  year={2022}
}

@article{li2022competition,
  title={Competition-level code generation with alphacode},
  author={Li, Yujia and Choi, David and Chung, Junyoung and Kushman, Nate and Schrittwieser, Julian and Leblond, R{\'e}mi and Eccles, Tom and Keeling, James and Gimeno, Felix and Dal Lago, Agustin and others},
  journal={Science},
  volume={378},
  number={6624},
  pages={1092--1097},
  year={2022},
  publisher={American Association for the Advancement of Science}
}

@online{Elastic,
  author="Elastic",
  title="Elastic",
  url=" https://www.elastic.co/enterprise-search/vector-search",
  note="(2024)",
}

@article{mohoney2023high,
  title={High-throughput vector similarity search in knowledge graphs},
  author={Mohoney, Jason and Pacaci, Anil and Chowdhury, Shihabur Rahman and Mousavi, Ali and Ilyas, Ihab F and Minhas, Umar Farooq and Pound, Jeffrey and Rekatsinas, Theodoros},
  journal={Proceedings of the ACM on Management of Data},
  volume={1},
  number={2},
  pages={1--25},
  year={2023},
  publisher={ACM New York, NY, USA}
}

@online{pgvector,
  author="pgvector",
  title="pgvector",
  url="https://github.com/pgvector/pgvector",
  note="(2024)",
}

@online{Qdrant,
  author="Qdrant",
  title="Qdrant",
  url="https://qdrant.tech/",
  note="(2024)",
}

@online{SingleStore,
  author="SingleStore",
  title="SingleStore",
  url="https://www.singlestore.com/built-in-vector",
  note="(2024)",
}

@inproceedings{wang2021milvus,
  title={Milvus: A purpose-built vector data management system},
  author={Wang, Jianguo and Yi, Xiaomeng and Guo, Rentong and Jin, Hai and Xu, Peng and Li, Shengjun and Wang, Xiangyu and Guo, Xiangzhou and Li, Chengming and Xu, Xiaohai and others},
  booktitle={Proceedings of the 2021 International Conference on Management of Data},
  pages={2614--2627},
  year={2021}
}

@inproceedings{yang2020pase,
  title={Pase: Postgresql ultra-high-dimensional approximate nearest neighbor search extension},
  author={Yang, Wen and Li, Tao and Fang, Gai and Wei, Hong},
  booktitle={Proceedings of the 2020 ACM SIGMOD international conference on management of data},
  pages={2241--2253},
  year={2020}
}

@article{lewis2020retrieval,
  title={Retrieval-augmented generation for knowledge-intensive nlp tasks},
  author={Lewis, Patrick and Perez, Ethan and Piktus, Aleksandra and Petroni, Fabio and Karpukhin, Vladimir and Goyal, Naman and K{\"u}ttler, Heinrich and Lewis, Mike and Yih, Wen-tau and Rockt{\"a}schel, Tim and others},
  journal={Advances in neural information processing systems},
  volume={33},
  pages={9459--9474},
  year={2020}
}

@inproceedings{huang2020embedding,
  title={Embedding-based retrieval in facebook search},
  author={Huang, Jui-Ting and Sharma, Ashish and Sun, Shuying and Xia, Li and Zhang, David and Pronin, Philip and Padmanabhan, Janani and Ottaviano, Giuseppe and Yang, Linjun},
  booktitle={Proceedings of the 26th ACM SIGKDD International Conference on Knowledge Discovery \& Data Mining},
  pages={2553--2561},
  year={2020}
}

@article{xiong2020approximate,
  title={Approximate nearest neighbor negative contrastive learning for dense text retrieval},
  author={Xiong, Lee and Xiong, Chenyan and Li, Ye and Tang, Kwok-Fung and Liu, Jialin and Bennett, Paul and Ahmed, Junaid and Overwijk, Arnold},
  journal={arXiv preprint arXiv:2007.00808},
  year={2020}
}

@inproceedings{babenko2014additive,
  title={Additive quantization for extreme vector compression},
  author={Babenko, Artem and Lempitsky, Victor},
  booktitle={Proceedings of the IEEE Conference on Computer Vision and Pattern Recognition},
  pages={931--938},
  year={2014}
}

@inproceedings{ge2013optimized,
  title={Optimized product quantization for approximate nearest neighbor search},
  author={Ge, Tiezheng and He, Kaiming and Ke, Qifa and Sun, Jian},
  booktitle={Proceedings of the IEEE conference on computer vision and pattern recognition},
  pages={2946--2953},
  year={2013}
}

@article{jegou2010product,
  title={Product quantization for nearest neighbor search},
  author={Jegou, Herve and Douze, Matthijs and Schmid, Cordelia},
  journal={IEEE transactions on pattern analysis and machine intelligence},
  volume={33},
  number={1},
  pages={117--128},
  year={2010},
  publisher={IEEE}
}

@inproceedings{martinez2018lsq++,
  title={LSQ++: Lower running time and higher recall in multi-codebook quantization},
  author={Martinez, Julieta and Zakhmi, Shobhit and Hoos, Holger H and Little, James J},
  booktitle={Proceedings of the European conference on computer vision (ECCV)},
  pages={491--506},
  year={2018}
}

@article{wang2017survey,
  title={A survey on learning to hash},
  author={Wang, Jingdong and Zhang, Ting and Sebe, Nicu and Shen, Heng Tao and others},
  journal={IEEE transactions on pattern analysis and machine intelligence},
  volume={40},
  number={4},
  pages={769--790},
  year={2017},
  publisher={IEEE}
}

@inproceedings{kalantidis2014locally,
  title={Locally optimized product quantization for approximate nearest neighbor search},
  author={Kalantidis, Yannis and Avrithis, Yannis},
  booktitle={Proceedings of the IEEE conference on computer vision and pattern recognition},
  pages={2321--2328},
  year={2014}
}

@article{chen2010approximate,
  title={Approximate nearest neighbor search by residual vector quantization},
  author={Chen, Yongjian and Guan, Tao and Wang, Cheng},
  journal={Sensors},
  volume={10},
  number={12},
  pages={11259--11273},
  year={2010},
  publisher={Molecular Diversity Preservation International (MDPI)}
}

@inproceedings{martinez2016revisiting,
  title={Revisiting additive quantization},
  author={Martinez, Julieta and Clement, Joris and Hoos, Holger H and Little, James J},
  booktitle={Computer Vision--ECCV 2016: 14th European Conference, Amsterdam, The Netherlands, October 11-14, 2016, Proceedings, Part II 14},
  pages={137--153},
  year={2016},
  organization={Springer}
}

@misc{foote2003treeq,
  title={TreeQ Manual V0. 8},
  author={Foote, Jonathan T},
  year={2003},
  publisher={September}
}

@inproceedings{charami2007performance,
  title={Performance evaluation of TreeQ and LVQ classifiers for music information retrieval},
  author={Charami, Matina and Halloush, Rami and Tsekeridou, Sofia},
  booktitle={IFIP International Conference on Artificial Intelligence Applications and Innovations},
  pages={331--338},
  year={2007},
  organization={Springer}
}

@article{liu2016query,
  title={Query-adaptive hash code ranking for large-scale multi-view visual search},
  author={Liu, Xianglong and Huang, Lei and Deng, Cheng and Lang, Bo and Tao, Dacheng},
  journal={IEEE Transactions on Image Processing},
  volume={25},
  number={10},
  pages={4514--4524},
  year={2016},
  publisher={IEEE}
}

@inproceedings{guo2016quantization,
  title={Quantization based fast inner product search},
  author={Guo, Ruiqi and Kumar, Sanjiv and Choromanski, Krzysztof and Simcha, David},
  booktitle={Artificial intelligence and statistics},
  pages={482--490},
  year={2016},
  organization={PMLR}
}

@article{wright2015coordinate,
  title={Coordinate descent algorithms},
  author={Wright, Stephen J},
  journal={Mathematical programming},
  volume={151},
  number={1},
  pages={3--34},
  year={2015},
  publisher={Springer}
}

@article{wang2018composite,
  title={Composite quantization},
  author={Wang, Jingdong and Zhang, Ting},
  journal={IEEE transactions on pattern analysis and machine intelligence},
  volume={41},
  number={6},
  pages={1308--1322},
  year={2018},
  publisher={IEEE}
}

@inproceedings{andre2016cache,
  title={Cache locality is not enough: High-performance nearest neighbor search with product quantization fast scan},
  author={Andr{\'e}, Fabien and Kermarrec, Anne-Marie and Le Scouarnec, Nicolas},
  booktitle={42nd International Conference on Very Large Data Bases},
  volume={9},
  number={4},
  pages={12},
  year={2016}
}

@article{sun2023soar,
  title={SOAR: improved indexing for approximate nearest neighbor search},
  author={Sun, Philip and Simcha, David and Dopson, Dave and Guo, Ruiqi and Kumar, Sanjiv},
  journal={Advances in Neural Information Processing Systems},
  volume={36},
  pages={3189--3204},
  year={2023}
}

@inproceedings{li2017fexipro,
  title={FEXIPRO: fast and exact inner product retrieval in recommender systems},
  author={Li, Hui and Chan, Tsz Nam and Yiu, Man Lung and Mamoulis, Nikos},
  booktitle={Proceedings of the 2017 ACM International Conference on Management of Data},
  pages={835--850},
  year={2017}
}

@article{deng2024efficient,
  title={Efficient Data-aware Distance Comparison Operations for High-Dimensional Approximate Nearest Neighbor Search},
  author={Deng, Liwei and Chen, Penghao and Zeng, Ximu and Wang, Tianfu and Zhao, Yan and Zheng, Kai},
  journal={arXiv preprint arXiv:2411.17229},
  year={2024}
}

@article{tepper2023leanvec,
  title={LeanVec: Searching vectors faster by making them fit},
  author={Tepper, Mariano and Bhati, Ishwar Singh and Aguerrebere, Cecilia and Hildebrand, Mark and Willke, Ted},
  journal={arXiv preprint arXiv:2312.16335},
  year={2023}
}

@article{tepper2024gleanvec,
  title={GleanVec: Accelerating vector search with minimalist nonlinear dimensionality reduction},
  author={Tepper, Mariano and Bhati, Ishwar Singh and Aguerrebere, Cecilia and Willke, Ted},
  journal={arXiv preprint arXiv:2410.22347},
  year={2024}
}

@article{malkov2018efficient,
  title={Efficient and robust approximate nearest neighbor search using hierarchical navigable small world graphs},
  author={Malkov, Yu A and Yashunin, Dmitry A},
  journal={IEEE transactions on pattern analysis and machine intelligence},
  volume={42},
  number={4},
  pages={824--836},
  year={2018},
  publisher={IEEE}
}

@article{HNSW,
  title={Efficient and robust approximate nearest neighbor search using hierarchical navigable small world graphs},
  author={Malkov, Yu A and Yashunin, Dmitry A},
  journal={IEEE transactions on pattern analysis and machine intelligence},
  volume={42},
  number={4},
  pages={824--836},
  year={2018},
  publisher={IEEE}
}

@article{NSG,
  title={Fast approximate nearest neighbor search with the navigating spreading-out graph},
  author={Fu, Cong and Xiang, Chao and Wang, Changxu and Cai, Deng},
  journal={arXiv preprint arXiv:1707.00143},
  year={2017}
}

@String{Computing = "Computing" }

@String{Computer = "{IEEE} Computer" }

@String{Springer = "Springer-Verlag" }

@ArtifactSoftware{R,
    title = {R: A Language and Environment for Statistical Computing},
    author = {{R Core Team}},
    organization = {R Foundation for Statistical Computing},
    address = {Vienna, Austria},
    year = {2019},
    url = {https://www.R-project.org/},
}

\appendix









\end{document}